\documentclass[aps,pra,twocolumn,letterpaper,reprint,nofootinbib]{revtex4-1}
\pdfoutput=1
\usepackage[utf8]{inputenc}
\usepackage[T1]{fontenc}
\usepackage{amsmath, amsthm, amssymb, mathtools, dsfont}
\usepackage{newtxtext} %
\usepackage[subscriptcorrection,smallerops]{newtxmath} %
\usepackage{MnSymbol} 
\usepackage{graphicx}
\usepackage{bbm}
\usepackage[normalem]{ulem} 
\usepackage[dvipsnames]{xcolor}
\usepackage{enumitem} 
\usepackage{array} 
\usepackage{booktabs} 
\usepackage[export]{adjustbox}
\usepackage{lipsum}

\usepackage{epstopdf}

\usepackage[colorlinks=true,linkcolor=blue,citecolor=blue,urlcolor=blue]{hyperref}
\usepackage{geometry}
\makeatletter 
\theoremstyle{definition}
\newtheorem{definition}{Definition}
\theoremstyle{definition}

\theoremstyle{definition}

\theoremstyle{definition}
\newtheorem{proposition}{Proposition}
\theoremstyle{definition}
\newtheorem{observation}{Observation}

\renewcommand\H{\mathcal{H}}
\newcommand{\1}{\mathds{1}}
\newcommand{\defeq}{\coloneqq}
\newcommand{\rdefeq}{\eqqcolon}
\newcommand{\ot}{\otimes}
\newcommand{\op}{\oplus}

\newcommand{\bra}[1]{\langle #1 \rvert}
\newcommand{\ket}[1]{\lvert #1\rangle}

\newcommand{\ketbra}[2]{\lvert {#1}\rangle \langle {#2}\rvert}

\newcommand{\proj}[1]{\ketbra{#1}{#1}}
\newcommand{\modu}[1]{\lvert #1 \rvert}

\DeclareMathOperator{\tr}{tr}

\DeclareMathOperator{\im}{im}

\DeclareMathOperator{\rank}{rank}
\newcommand{\ncl}{\nonumber\\}

\newcommand*{\mtiny}{\scriptscriptstyle}
\newcommand*{\inv}{^{-1}}

\newcommand*{\ad}{^{\dagger}}

\newcommand{\Rc}{\mathcal R}
\newcommand{\Nc}{\mathcal N}

\newcommand{\Sc}{\mathcal S}
\newcommand{\Qc}{\mathcal Q}

\newcommand{\mbi}{{\operatorname{MBI}}}
\newcommand{\mpi}{{\operatorname{MPI}}}
\newcommand{\bic}{{\operatorname{BI}}}
\newcommand{\pic}{{\operatorname{PI}}}

\newcommand{\E}{\textbf{E}}
\newcommand{\F}{\textbf{F}}
\renewcommand{\P}{\textbf{P}}
\newcommand{\Ec}{\mathcal E}
\newcommand{\Ect}{{\mtiny\Ec}}

\newcommand{\Ic}{\mathcal{I}}
\newcommand{\Mc}{\mathcal{M}}
\newcommand{\Uc}{\mathcal{U}}
\newcommand{\Tc}{\mathcal T}
\newcommand{\id}{\operatorname{id}}
\newcommand{\rel}{\operatorname{rel}}
\newcommand{\rob}{\operatorname{rob}}
\newcommand{\geo}{\operatorname{geo}}
\newcommand{\He}{\H_\Ect}
\newcommand{\Pie}{\Pi_\Ect}

\newcommand{\trn}[1]{\lvert\lvert #1 \rvert\rvert_1}

\makeatother 

\begin{document}
\title{Quantifying coherence with respect to general quantum measurements}
\author{Felix Bischof}
\email{felix.bischof@hhu.de}
\author{Hermann Kampermann}\author{Dagmar Bru\ss}
\affiliation{Institut f\"ur Theoretische Physik III, Heinrich-Heine-Universit\"at D\"usseldorf,
Universit\"atsstra\ss e 1, D-40225 D\"usseldorf, Germany}
\date{\today}

\begin{abstract}
Coherence is a cornerstone of quantum theory and a prerequisite for the advantage of quantum technologies. In recent work, the notion of coherence with respect to a general quantum measurement (POVM) was introduced and embedded into a resource-theoretic framework that generalizes the standard resource theory of coherence. In particular, POVM-incoherent (free) states and operations were established. In this work,
we explore features of this framework which arise due to the rich structure of POVMs compared to projective measurements. Moreover, we introduce a rigorous, probabilisitic framework for POVM-based coherence measures and free operations. This leads to the introduction of new, strongly monotonic resource measures that neatly generalize well-known standard coherence measures. Finally, we show that the relative entropy of POVM-coherence is equal to the cryptographic randomness gain, providing an important operational meaning to the concept of coherence with respect to a general measurement.
\end{abstract}
\maketitle

\section{Introduction}
In quantum technologies, particular properties of quantum states and channels become valuable resources for the application. For example, quantum entanglement enables superior performance in nonlocal games compared to classical resources, which can be utilized for the device-independent distribution of a secret key~\cite{acin2007device,arnon2018practical}. Quantum resource theories (QRTs)~\cite{brandao2015reversible,liu2017resource,chitambar2018quantum} provide a versatile, application-independent methodology for the quantitative analysis of resources. The QRT framework has been applied to different quantum phenomena such as entanglement~\cite{horodecki2003local,horodecki2009quantum}, purity~\cite{horodecki2003reversible}, asymmetry~\cite{marvian2013theory,marvian2014extending}, thermodynamics~\cite{brandao2013resource} and coherence~\cite{baumgratz2014quantifying,winter2016operational,streltsov2016quantum}.
In recent years, the core common structure of QRTs has been identified~\cite{coecke2016mathematical,horodecki2013quantumness}. In physical setups, the feasible quantum operations are usually constrained, either due to practical limitations or  fundamental physical laws such as energy conservation. Consequently, only a subclass of operations can be (easily) realized, which are called free operations. Properties of quantum states that cannot be created by free operations are considered a resource. States without resource content are called free states. Building on these basic notions, it is possible to develop a rigorous quantitative framework which yields insights into the different means of quantifying a resource, the optimal distillation and dilution of the resource and the possibility of interconversion of resource states under the given constraints.

Quantum coherence~\cite{streltsov2016quantum}, i.e., the feature of quantum systems to be in a superposition of different states is at the core of quantum mechanics. In particular, coherence underlies quantum entanglement~\cite{streltsov2015measuring} which plays a central role in quantum communication and computing. The resource theory of coherence is formulated with respect to a distinguished basis of a Hilbert space, the incoherent basis $\{\ket{i}\}$, which defines free states as the states that are diagonal in this basis. For instance, in quantum thermodynamics $\{\ket{i}\}$ is the energy eigenbasis and work can be extracted by a thermal process which removes the off-diagonal entries of the state of the system~\cite{kwon2018clock}. Equivalently, coherence can be defined with respect to the von Neumann measurement $\P=\{\proj{i}\}$ such that free states arise as post-measurement states of $\P$.

However, coherence as an intrinsic property of quantum states should be defined with respect to the most general quantum measurements, namely, positive-operator-valued measures (POVMs). This is because POVMs describe the most general type of quantum observable and can have a real operational advantage compared to any projective measurement, see e.g.~\cite{oszmaniec2017simulating}.
A notion of \emph{coherence with respect to a general measurement} is meaningful if i) it can be embedded in a consistent resource theory ii) POVM-based coherence measures have interesting operational interpretations, i.e, they quantify the advantage of states in a quantum information protocol.
Recently, a resource theory of quantum state coherence with respect to an arbitrary POVM was introduced and studied~\cite{bischof2018resource}. Here, we develop this framework further by discussing selected features that are distinct from standard coherence theory. In particular, we answer point ii) by providing an important operational interpretation of the most fundamental POVM-coherence measure. Moreover, we introduce further operational restrictions on the class of free operations in conjunction with new useful measures of POVM-coherence. We expect that our findings will help to clarify the role of coherence in all quantum technologies employing nonprojective measurements.

The structure of our work is as follows. In Sec.~\ref{sec:pbcoh} we briefly recapitulate the resource theory of POVM-based coherence~\cite{bischof2018resource}. Sec.~\ref{sec:mincoh} discusses a particular one-parameter POVM, which describes how standard coherence turns into POVM-based coherence, highlighting features of minimally coherent states and the measurement map. In Sec.~\ref{sec:cohran}, we show that the relative entropy of POVM-based coherence quantifies the cryptographic randomness of the measurement outcomes in relation to an eavesdropper who has side information about the measured state. This provides an operational interpretation of the resource theory. Subsequently, in Sec.~\ref{sec:selpbc}, we define and study free Kraus operators as well as selective free operations. Finally, in Sec.~\ref{sec:robcoh}, we introduce new, strongly monotonic POVM-coherence measures and find relations among them.

\subsection{Resource theory of block coherence}
\label{sec:bbcoh}

The resource theory of POVM-based coherence is derived from the framework of block coherence\footnote{{In \r{A}berg's work block coherence is called superposition. However, since block coherence is a generalization of coherence with very similar structure, we find this name more suitable from the current literature perspective.}}, introduced by \r{A}berg~\cite{aberg2006quantifying}. In the latter resource theory, the Hilbert space $\H=\op_i\pi_i$ is partitioned into orthogonal subspaces $\pi_i$. If we denote the projector on the $i$-th subspace by $P_i$, the set $\P=\{P_i\}$ constitutes a projective measurement on $\H$. Block-incoherent (BI, free) states are defined as states of the form
\begin{align}
\rho_{\mtiny\bic} &= \Delta[\sigma], \quad\sigma\in\Sc \label{bic}, \\
\Delta[\sigma]& = \sum_iP_i\sigma P_i \label{bdeph},
\end{align}
where $\Sc$ is the set of quantum states and $\Delta$ denotes the block-dephasing operation, which sets all entries except the blocks on the diagonal to zero. In other words, block-incoherent states do not possess ``outer'' coherence across the subspaces $\pi_i$. Note that the convex set of block-incoherent states $\Ic$ is equal to the set of U(1)-symmetric states in the resource theory of asymmetry with the symmetry group $\{U(\theta)=e^{-i\theta\sum_kkP_k}\}$~\cite{piani2016robustness}. 
A further ingredient of the resource theory are maximally block-incoherent (MBI) operations $\Lambda_\mbi$. These are channels (i.e., completely positive trace-preserving maps) that preserve the set of block-incoherent states\footnote{In the resource theory of asymmetry, the free operations usually considered in the literature~\cite{marvian2013theory,marvian2014extending,gour2009measuring,marvian2016quantum} are the group-covariant operations, i.e., channels that commute with all unitary channels obtained from the symmetry group. In the language of coherence theory, these operations are the translationally-invariant operations~\cite{marvian2016quantum}, which form a strict subset of the maximal set of free operations MBI we consider here~\cite{marvian2016quantify}.}, that is, $\Lambda_\mbi[\Ic]\subseteq\Ic$.
Finally, the block-coherence content of states can be quantified by suitable measures~\cite{aberg2006quantifying}. The standard example for a measure is the relative entropy of block coherence, which has the form
\begin{align}\label{rebc}
C_{\rel}(\rho,\P)=S(\Delta[\rho])-S(\rho),
\end{align}
where $S$ denotes the von Neumann entropy $S(\rho)=-\tr(\rho\log_2\rho)$. The quantity $C_{\rel}$ satisfies the following properties which we view as minimal requirements for a block-coherence measure~\cite{bischof2018resource}:
\begin{enumerate}
\label{bcmeasp}
\item[(B1)] \emph{Faithfulness:} $C(\rho,\P)\geq0$ with equality iff $\rho=\rho_{\mtiny\bic}$.
\item[(B2)] \emph{Monotonicity:} $C(\Lambda_\mbi[\rho],\P)\leq C(\rho,\P)$ for any MBI map. 
\item[(B3)] \emph{Convexity:} $C(\sum_ip_i\rho_i,\P)\leq\sum_ip_iC(\rho_i,\P)$ for all states $\{\rho_i\}$, and probabilities $p_i\geq0$, $\sum_ip_i=1$.
\end{enumerate}
Note that the concepts explained so far coincide with their counterparts in the standard resource theory of coherence if all $P_i$ have rank one.

\subsection{Resource theory of coherence based on POVMs}
\label{sec:pbcoh}
A much broader generalization of standard coherence is provided by the POVM-based resource theory of coherence~\cite{bischof2018resource}. POVMs describe the most general type of quantum measurement, namely a collection of $n$ positive operators $\E=\{E_i\geq0\}_{i=1}^n$ that sum to the identity, $\sum_iE_i=\1$. We will also use the corresponding measurement operators, defined as $A_i=U_i\sqrt{E_i}$. Here, $\sqrt{E_i}$ denotes the unique positive square root of $E_i$ and $U_i$ is an arbitrary unitary. Thus,
$A_i\ad A_i=E_i$ holds.

Let $\E$ be a POVM on a $d$-dimensional Hilbert space $\H$. The main idea to define POVM-based coherence theory is to link it to the resource theory of block coherence specified by the \emph{Naimark extension} $\P$ of $\E$. The Naimark extension is a projective measurement with the following property: if the POVM is embedded into a subspace of a higher-dimensional Hilbert space $\H'$ of suitable dimension $d'\geq d$, $\P$ extends $\E$ to the whole space. We denote by $\Ec$ an (isometric) embedding channel, mapping operators on $\H$ to operators on $\H'$. Consequently, it holds that
\begin{align}
\tr(E_i\rho) = \tr(P_i\Ec[\rho]),\quad\textrm{for all $\rho\in\Sc$,}
\end{align}
that is, $\P$ has the same expectation values for any embedded state $\Ec[\rho]$ as $\E$ for $\rho$. Therefore, it is natural to define the coherence of a state $\rho$ w.r.t.~a POVM $\E$ as the block coherence of $\Ec[\rho]$ w.r.t.~the Naimark extension $\P$ of $\E$, namely
\begin{align}\label{pbcmeas}
C(\rho,\E)\defeq C(\Ec[\rho],\P),
\end{align}
where the function $C$ on the right denotes any unitarily-covariant block-coherence measure~\cite{bischof2018resource}. Note that the Naimark extension of a POVM $\E$, in particular its dimension $d'$, is not unique\footnote{For instance, given any Naimark extension, one can always increase the dimension of each effect by adding projections on additional degrees of freedom.}. Therefore, one should ensure that the right side of Eq.~\eqref{pbcmeas} does not depend on the choice of Naimark extension $\P$. This property was shown in~\cite{bischof2018resource} for the case of $C(\rho',\P)=C_{\rel}(\rho',\P)$ from Eq.~\eqref{rebc}. One obtains the relative entropy of POVM-based coherence
\begin{align}\label{repbc}
C_{\rel}(\rho,\E)=H(\{p_i(\rho)\})+\sum_ip_i(\rho)S(\rho_i)-S(\rho),
\end{align}
with $p_i(\rho)=\tr(E_i\rho)$, $\rho_i=A_i\rho A_i\ad/p_i$, $A_i=\sqrt{E_i}$, and the Shannon entropy $H(\{p_i(\rho)\})=-\sum_ip_i\log_2 p_i$. In the special case of $\E$ being a von Neumann measurement, $E_i=\proj{i}$, $C_{\rel}(\rho,\E)$ corresponds to the standard relative entropy of coherence. From Def.~\eqref{pbcmeas} it follows that for some POVMs the set of states with zero coherence (POVM-incoherent states $\rho_{\mtiny\pic}$) is empty~\cite{bischof2018resource}. The generalization of incoherent states are states with \emph{minimal} coherence $\rho_{\min}$, which form a set $\Mc$ that has similar properties as the standard incoherent set: it is nonempty, convex, and closed under POVM-incoherent operations, which are defined below.

POVM-incoherent (free) operations can be derived from block-incoherent operations on the enlarged space. Let $\Lambda'_\mbi$ be a block-incoherent map on states $\rho'\in\Sc'$ on the Naimark space with the additional property that the set of embedded states $\{\Ec[\rho]\in\Sc':\rho\in\Sc\}$ is closed under $\Lambda'_\mbi$. Then, the following channel is called a (maximally) POVM-incoherent operation (MPI)~\cite{bischof2018resource}
\begin{align}\label{picops}
\Lambda_\mpi[\rho]  = \Ec\inv\circ \Lambda'_\mbi \circ\Ec[\rho].
\end{align}
POVM-coherence measures and MPI maps are the main constituents of the resource theory of quantum state coherence based on POVMs. Crucially, these two concepts are consistent with each other by construction, as any POVM-based coherence measure~\eqref{pbcmeas} satisfies:
\begin{enumerate}
\item[(P1)] \emph{Faithfulness:} $C(\rho,\E)\geq0$ with equality iff $\rho=\rho_{\mtiny\pic}$.
\item[(P2)] \emph{Monotonicity:} $C(\Lambda_\mpi[\rho],\E)\leq C(\rho,\E)$ for any MPI map with respect to~$\E$. 
\item[(P3)] \emph{Convexity:} $C(\rho,\E)$ is convex in $\rho$.
\end{enumerate}
See Ref.~\cite{bischof2018resource} for a detailed discussion of the concepts. The question whether POVM-coherence measures satisfy \emph{strong} monotonicity is an open problem that will be addressed and answered in Sec~\ref{sec:selpbc}.

\section{Minimally coherent states and the measurement map}
\label{sec:mincoh}

In this section, we examine a one-parameter POVM to illustrate how standard coherence theory turns into POVM-based coherence. Moreover, this example sheds light on two natural questions in the context of the generalized notion of coherence: i) does the maximally mixed state always contain the lowest amount of coherence? ii) is the measurement map $\Lambda_\E[\rho]=\sum_i\sqrt{E_i}\rho \sqrt{E_i}$ POVM-incoherent for any POVM? In standard coherence theory, both questions can be answered in the affirmative. However, our example shows that this does not hold in general.

To illustrate the amount of POVM-based coherence in states, we discuss a POVM representing the continuous distortion from a von Neumann measurement into a non-projective POVM. Concretely, we consider $\E(\delta)=\{E_i(\delta)\}_{i=1}^3$ which coincides for $\delta=0$ with the qubit $Y$-measurement, and for $\delta=1$ with the qubit trine POVM, whose measurement directions $\vec m_i$ form an equilateral triangle on the xy-plane of the Bloch sphere. With the Bloch representation of qubit POVMs
\begin{align}
E_i&=\alpha_i(\1+\vec m_i\cdot\vec\sigma)\quad\textrm{with}\quad \alpha_i\geq0 \ncl
\sum_i\alpha_i&=1,\quad \sum_i\alpha_i\vec m_i=0,
\end{align}
the POVM elements $E_i(\delta)$ are given by the parameters
\begin{align}\label{aux5}
&\alpha_1=\frac\delta3,\quad \alpha_2=\alpha_3=\frac12(1-\frac\delta3) \ncl
&\vec m_1 = (1,0,0)^T\quad\textrm{and with}\quad t\defeq\frac\delta{3-\delta} \ncl
&\vec m_2 = (-t,\sqrt{1-t^2},0)^T\ncl
&\vec m_3 = (-t,-\sqrt{1-t^2},0)^T.
\end{align}
The effects $E_i(\delta)$ are linearly independent (except for $\delta=0$) as the measurement directions form a triangle~\cite{d2005classical}. Moreover, since $\modu{\vec m_i}=1$, the effects have rank one, except for $\delta=0$ where the first effect has rank zero. Thus, $\E(\delta)$ is an \emph{extremal} POVM for any $\delta$, i.e., it cannot be written as a mixture of two other POVMs, and in this sense does not contain classical noise. 
In Fig.~\ref{disty}, we plot the POVM-based coherence of selected states, as well as the minimally and maximally achievable coherence for all values of $\delta$. Interestingly, the figure shows that for $0<\delta<1$, the state with minimal coherence is distinct from the maximally mixed state. We abstain from stating the explicit form of $\rho_{\min}(\delta)$ in the range $0<\delta<1$ as it is too cumbersome. However, we report that in this interval the maximal eigenvalue takes values $0.5<\lvert\lvert\rho_{\min}(\delta)\rvert\rvert_\infty \lesssim0.6$. 

\begin{figure}[!ht]
\centering
\includegraphics[width=0.85\columnwidth]{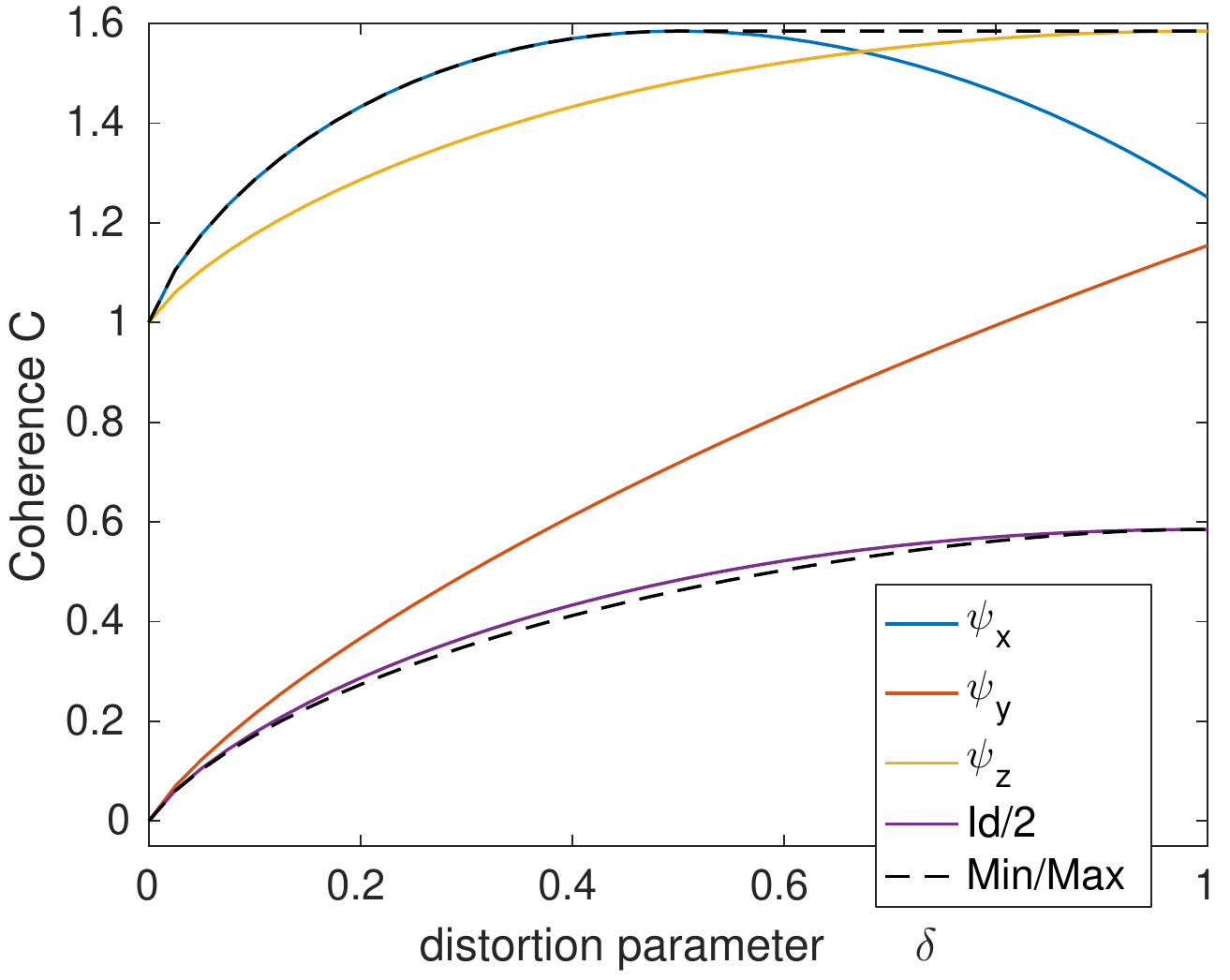} 
\caption{The relative entropy of POVM-based coherence plotted for selected states with respect to the POVM $\E(\delta)$ defined in Eq.~(\ref{aux5}) for all values of the distortion parameter $\delta$. The states $\psi_{x},\psi_{y},\psi_{z}$ denote the $+1$-eigenstates of the Pauli matrices $\sigma_{x}, \sigma_{y}, \sigma_{z}$, respectively. The lowest solid line corresponds to the maximally mixed state.
The dashed lines indicate the achievable minimal and maximal coherence, respectively, which were obtained analytically (by Karush-Kuhn-Tucker conditions~\cite{bischof2018resource}).
\label{disty}}
\end{figure}

This property can be utilized to show that the measurement map of the POVM $\E$, defined as
\begin{align}\label{povmmap}
\Lambda_\E [\rho] = \sum_i \sqrt{E_i}\rho\sqrt{E_i},
\end{align}
which is unital, $\Lambda_\E [\1]=\1$, is not incoherent in general. A counterexample is provided by the POVM $\E(\delta)$: Table~\ref{tab:disty} shows for selected parameters of $\delta$ that $\Lambda_\E$ increases the coherence of $\rho_{\min}$ for $0<\delta<1$. However, note that $\Lambda_\E$ from Eq.~\eqref{povmmap} is POVM-incoherent for any projective measurement but also for certain nonprojective measurements like the qubit trine POVM~\cite{bischof2018resource}.

\begin{table}[!h]
  \centering
\begin{tabular}{||c|c|c|c||} 
 \hline
 $\delta$ & $C_{\rel}(\rho_{\min})$ 
 & $C_{\rel}(\Lambda_\E[\rho_{\min}])$ & $C_{\rel}(\1/2)$ \\ [0.5ex] 
 \hline\hline
 0 & 0 & 0 & 0 \\ \hline
 0.4 & 0.412 & 0.427 & 0.433 \\ \hline
 0.5 & 0.462 & 0.476 & 0.483 \\ \hline
 0.6 & 0.503 & 0.514 & 0.522 \\ \hline
 1 & 0.585 & 0.585 & 0.585 \\ [0.1ex] 
 \hline
\end{tabular}
\caption{POVM-based coherence of states w.r.t.~$\E(\delta)$ for selected values of $\delta$. For $\delta\in\{0,1\}$, the maximally mixed state $\1/2$ is a state $\rho_{\min}$ of minimal coherence. Moreover, the measurement map $\Lambda_\E$ is incoherent in these cases and thus does not increase the coherence of $\rho_{\min}$. 
For $0<\delta<1$,  the maximally mixed state $\1/2$ does not have minimal coherence and $\Lambda_\E$ increases the coherence of $\rho_{\min}$.
}
\label{tab:disty}
\end{table}

\section{POVM-based coherence and private randomness}
\label{sec:cohran}

In Ref.~\cite{bischof2018resource}, the relative entropy of POVM-based coherence $C_{\rel}(\rho,\E)$ from Eq.~\eqref{repbc} was established as a measure of coherence with respect to general measurements. However, in the previous work the operational meaning of this measure was left open. In this section, we show that $C_{\rel}(\rho,\E)$ quantifies the private randomness generated by the POVM $\E$ on the state $\rho$ with respect to an eavesdropper holding optimal side information about the measured state. 
This is a relevant result for quantum randomness generation and cryptography, which generalizes the findings from Refs.~\cite{yuan2015intrinsic,yuan2016interplay}, where it was shown that the standard relative entropy of coherence corresponds to the quantum randomness of a von Neumann measurement.

We consider a POVM $\F=\{F_i\}$ that is measured on a state $\rho_A$ on a quantum system $A$, such that the measurement outcomes $i$ are stored in the register $X$, see Fig.~\ref{fig:randomcoh}. An eavesdropper holds maximal side information about $\rho_A$, i.e., all degrees of freedom correlated with $A$ in the form of a purifying system $E$ such that $\ket{\psi}_{AE}$ with $\rho_A=\tr_E(\proj{\psi}_{AE})$ describes the joint pure state. After the measurement $\F$, the joint state is given by
\begin{align}\label{pstate}
\tilde\rho_{XAE} = \sum_ip_i\proj{i}_X\ot\proj{\tilde\psi_i}_{AE},
\end{align}
where $p_i=\tr(F_i\rho_A)$ denotes the probability to obtain outcome $i$. The pure post-measurement states $\ket{\tilde\psi_i}_{AE}=\frac1{\sqrt{p_i}}(A_i\ot\1)\ket{\psi}_{AE}$ are defined by the measurement operators $A_i$ that implement the POVM, that is, $F_i=A_i\ad A_i$.

Let $S(X|E)_\rho=S(\rho_{XE})-S(\rho_E)$ denote the conditional von Neumann entropy 
of $X$ given $E$ on the state $\rho$. We define the \emph{randomness} contained in the random variable $X=(i,p_i)$ of the measurement outcomes of $\F$ as
\begin{align}\label{rxe}
R_{X|E}(\rho_A)= \min_{\ket{\psi}_{AE}}S(X|E)_{\tilde\rho},
\end{align}
where $\tilde\rho=\tilde\rho_{XE}$ is obtained from Eq.~\eqref{pstate} by tracing out $A$ and the minimum is taken over all purifications $\ket{\psi}_{AE}$ of $\rho_A$.
This choice of randomness quantification is relevant in practice, as it describes the asymptotic private randomness, i.e., unpredictability of the measurement outcomes. Indeed, for an eavesdropper employing an independent and identically distributed (IID) attack in an $n$-round protocol, the single-round von Neumann entropy is related by the quantum asymptotic equipartition property~\cite{tomamichel2009fully} to the smooth quantum min-entropy $H_{\min}^\varepsilon(X^n|E^n)$ of all $n$ rounds. The latter quantity has been proven to quantify composable security in quantum randomness generation and cryptography. More precisely, $H_{\min}^\varepsilon(X^n|E^n)$ is equal to the minimal number of bits needed to reconstruct $X^n$ from $E^n$, except with probability of order $\varepsilon$~\cite{renner2008security,arnon2018practical}.

\begin{figure}[!ht]
\centering
\includegraphics[width=0.6\columnwidth,valign=T]{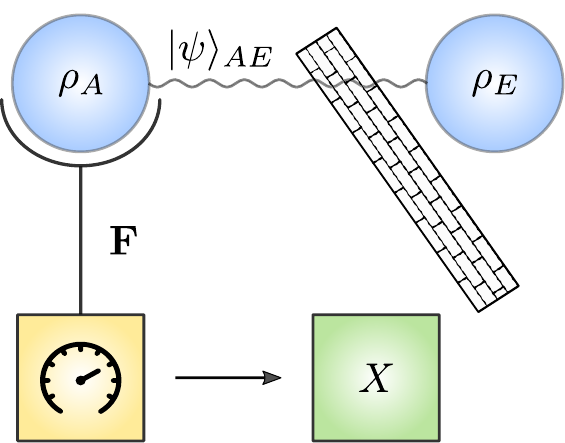}
\caption{\label{fig:randomcoh}
The relation between private randomness and POVM-based coherence. The eavesdropper Eve has maximal side information about the state $\rho_A$, namely a purification $\ket{\psi}_{AE}$. Nonetheless, if $\rho_A$ possesses coherence with respect to the POVM $\F$, the measurement outcomes $X=i$ contain secrecy with respect to Eve. That is, the asymptotic randomness generation rate is given by $R_{X|E}(\rho_A)=C_{\rel}(\rho_A,\F)$, with the relative entropy of POVM-based coherence defined in Eq.~\eqref{repbc}.
}
\end{figure}

\begin{proposition}
Let Eve hold a purification of $\rho_A$. The private randomness generation rate is equal to the relative entropy of POVM-based coherence, $R_{X|E}(\rho_A)=C_{\rel}(\rho_A,\F)$, for any possible POVM $\F$ measured on $\rho_A$ generating the outcome random variable $X$. 
\end{proposition}
\begin{proof} --
First, note that the local measurement $\F$ on $A$ leaves the state $\rho_E=\tr_A(\proj{\psi}_{AE})$ invariant, i.e., $\tilde\rho_E=\rho_E$. Moreover, it holds that $S(\rho_E) = S(\rho_A)$ since $\rho_{AE}=\proj{\psi}_{AE}$ is pure, and likewise $S(\tilde\rho_{A|i})=S(\tilde\rho_{E|i})$ since $\tilde\rho_{AE|i}=\proj{\tilde\psi_i}_{AE}$ is pure. 
This argument is a direct consequence of the Schmidt decomposition of pure states~\cite{nielsenchuang}.
Therefore, it holds that
\begin{align}
R_{X|E}(\rho_A) 
&= \min_{\ket{\psi}_{AE}} \bigl\{ S(\sum_ip_i\proj{i}_X\ot\tilde\rho_{E|i})-S(\tilde\rho_E) \bigr\} \ncl
&= \min_{\ket{\psi}_{AE}} \bigl\{ H(\{p_i\})+\sum_ip_iS(\tilde\rho_{E|i})-S(\rho_E) \bigr\} \ncl
&= H(\{p_i\})+\sum_ip_iS(\tilde\rho_{A|i})-S(\rho_A). 
\end{align}
In the first line, we inserted the state $\tilde\rho_{XE}$ from Eq.~\eqref{pstate} into Eq.~\eqref{rxe}.
In the second equation, we employed the joint entropy theorem~\cite{nielsenchuang}. 
The minimization can be dropped in the last step, as all quantities are independent of the choice of purification $\ket{\psi}_{AE}$. By inspecting Eq.~\eqref{repbc} we see that the expression in the last line is equal to $C_{\rel}(\rho_A,\F)$.
\end{proof}

This result explains why noisy POVMs typically lead to higher values of POVM-based coherence than projective measurements. The noise injects randomness into the outcomes $X$, which cannot be predicted by an eavesdropper with side information about the measured state. It is crucial that the eavesdropper does not have access to the measurement device, i.e., any noise in the measurement device is trusted. However, if the POVM $\E$ is extremal, the results of Ref.~\cite{ma15,bischof2017measurement,brask2017megahertz,ioannou2019upper} show that an eavesdropper cannot get additional knowledge about the measurement outcomes by pre-programming the measurement device. Extremal measurements such as the qubit trine POVM are thought to possess intrinsic quantum noise~\cite{d2005classical}, explaining why even the maximally mixed state can generate nonzero trusted randomness.
The POVM $\E(\delta)$ from Eq.~\eqref{aux5} is extremal for any $\delta\in[0,1]$. Thus, Fig.~\ref{disty} shows the generated private randomness $R_{X|E}(\rho)$ for selected states $\rho$ and the advantage of POVMs over projective measurements. In particular, for $\delta\geq\frac12$, $\E(\delta)$ yields up to $\log_2(3)\approx1.58$ private random bits per measurement, compared to maximally one bit for qubit projective measurements.

\section{Probabilistically free operations and strong monotonicity}
\label{sec:selpbc}

POVM-incoherent operations as defined in Eq.~\eqref{picops} form the set MPI, that is, the largest class of channels that cannot create POVM-based coherence. Thus, MPI generalizes the set of maximally-incoherent operations MIO~\cite{streltsov2016quantum}. However, in practice it is useful to also have a notion of \emph{selective} POVM-incoherent operations, which we introduce in this section. These operations cannot create coherence, not even probabilistically, when a particular outcome of the channel is selected. This stronger notion of incoherent operations was introduced in Ref.~\cite{baumgratz2014quantifying} for the standard resource theory of coherence under the name of incoherent operations (IO). It holds that incoherent operations are strictly included in the maximal set, IO $\subset$ MIO.

\subsection{Block-incoherent Kraus operators}
As a first building block, we need to introduce Kraus operators that cannot create block coherence. Let $\P$ be any projective measurement defining the Hilbert space partition $\H'=\op_i\pi_i$, where $\pi_i = \im P_i$. In Sec.~\ref{sec:pbcoh} we have introduced the block-dephasing operation $\Delta$ and block-incoherent states in Eq.~\eqref{bic}. Consequently, block-incoherent pure states are element of the set $\{\ket{\varphi_i}\}_i$, where $\ket{\varphi_i}$ denotes any normalized state vector such that 
\begin{align}\label{purebic}
\ket{\varphi_i}\in\im P_i.
\end{align}
Note that if $\dim P_i\geq2$, the above set is not finite as superpositions within $\im P_i$ are allowed.

Let $\{K_l'\}$ be a set of Kraus operators on $\H'$, that is, the operators satisfy the normalization condition $\sum_l(K_l')\ad K_l'=\1$. We call a Kraus operator block-incoherent if
\begin{align}\label{bickraus}
K'_l\ket{\varphi_i} \propto \ket{\varphi_j}
\end{align}
holds for all block-incoherent pure states $\ket{\varphi_i}$. Note that in analogy to the case in standard coherence theory~\cite{winter2016operational} block-incoherent Kraus operators have the form 
\begin{align}\label{aux10}
K_l' = \sum_{i}P_{f(i)}C_lP_{i},
\end{align}
where $f$ is some index function, which has to be chosen together with the complex matrix $C_l$ on $\H'$ such that normalization holds. We call a Kraus operator $K_l'$ \emph{strictly} block-incoherent, if $f$ is invertible, that is, an index permutation. In this case, also $(K_l')\ad$ is block-incoherent.

\subsection{POVM-incoherent Kraus operators}
\label{sec:pickraus}

Next, we construct Kraus operators that cannot create POVM-coherence in analogy to the construction of MPI operations~\eqref{picops}. 
We consider a POVM $\E$ on the $d$-dimensional space $\H$ and any Naimark extension $\P$ of it, defined on the $d'$-dimensional space $\H'$. The (Naimark) embedding of $\H$ into $\H'$ is given by $\H\op0\rdefeq\He$, which is a choice we make for the sake of concreteness without loss of generality. Define the operator 
\begin{align}\label{tmatrix}
T=\begin{pmatrix} \1 \\ 0 
\end{pmatrix},
\end{align}
where $0$ denotes the zero matrix of size $(d'-d)\times d$. Consequently, operators $X$ on $\H$ are transformed to Naimark space operators by the isometric channel $\Ec[X]=TX T\ad$. It holds that $T\ad T=\1$ and $TT\ad=\1\op0\rdefeq\Pie$.

Let $\{K_l'\}$ be a set of block-incoherent Kraus operators~\eqref{bickraus} on $\H'$, where any operator additionally satisfies
\begin{align}\label{spkraus}
K'_l\Pie = \Pie K'_l\Pie.
\end{align}
In other words, $K_l'$ maps the embedded original space $\H\op0$ to itself, which we call the subspace-preserving property. It is fulfilled if and only if all Kraus operators are of the form
\begin{align}\label{aux11}
K_l' = 
\left(\begin{array}{@{}c|c@{}}
* & * \\
\hline
0 & * \\
\end{array}\right),
\end{align}
where $0$ denotes the zero matrix of size $(d'-d)\times d$ and where $*$ represents matrices of suitable dimension.

\begin{definition}\label{def:pickraus}
We call the following operator on $\H$ a POVM-incoherent (PI) Kraus operator:
\begin{align}\label{pickraus}
K_l=T\ad K_l' T,
\end{align}
where $T$ is given in~\eqref{tmatrix} and $K_l'$ satisfies~\eqref{bickraus},~\eqref{spkraus} and normalization. 
\end{definition}
In Eq.~\eqref{pickraus}, the operators $T\ad$ and $T$ extract the upper left $d\times d$ block of the $d'\times d'$-matrix $K_l'$. One can readily check that a PI set $\{K_l\}$ satisfies normalization by construction.
At this point, we need to ensure that the above definition is not ambiguous.

\begin{proposition}\label{prop:pickraus}
The set containing all POVM-incoherent (PI) Kraus operators $K_l$ does not depend on the choice of Naimark extension used to define it, see Eq.~\eqref{pickraus}.
\end{proposition}

The proof can be found in the Appendix~\ref{app:rne}. In the special case of a von Neumann measurement, $\E$ can be chosen as its own Naimark extension such that $d'=d$. Thus, in this case Def.~\ref{def:pickraus} and Prop.~\ref{prop:pickraus} imply that PI Kraus operators are equivalent to standard incoherent Kraus operators.

\subsection{Selective free operations and strong monotonicity}
Building on the previous section, we are ready to define two classes of probabilistically free channels. These have the property that even when we post-select outcomes of the operation, POVM-coherence cannot be created from an incoherent input state. We call a channel $\Lambda$ a \emph{selective} POVM-incoherent (PI) operation, if it admits a Kraus decomposition $\Lambda[X]=\sum_lK_lXK_l\ad$ such that all operators $K_l$ are POVM-incoherent~\eqref{pickraus}. Moreover, we call $\Lambda$ \emph{strictly} POVM-incoherent (SPI), if additionally all adjoint operators $(K_l)\ad$ are POVM-incoherent. These definitions clearly generalize the classes of incoherent operations IO and strictly incoherent operations SIO~\cite{streltsov2016quantum}, respectively. We obtain the following hierarchy of POVM-incoherent operations
\begin{align}
\textrm{SPI } \subseteq \textrm{ PI } \subseteq \textrm{ MPI},
\end{align}
 where MPI denotes the maximal set of POVM-incoherent operations from Eq.~\eqref{picops}.

This leads to the following definition, which extends the requirements on a POVM-coherence measure $C(\rho,\E)$  from Sec.~\ref{sec:pbcoh}. It guarantees that free operations cannot create coherence on average when the observer has access to measurement results.
\begin{enumerate}
\item[(P2s)] \emph{Strong monotonicity of POVM-coherence measure:} $C(\rho,\E)$ does not increase on average under selective POVM-incoherent operations PI, i.e.,
\begin{align}
\sum_lp_lC(\rho_l,\E) \leq C(\rho,\E)
\end{align}
for any set of POVM-incoherent Kraus operators $K_l$ defining probabilities $p_l=\tr(K_l\rho K_l\ad)$ and post-measurement states $\rho_l=K_l\rho K_l\ad/p_l$.
\item[(B2s)] \emph{Strong monotonicity of block-coherence measure:} Same as (P2s) for the special case of projective measurements $\E=\P$ and selective block-incoherent operations BI.
\end{enumerate}

Note that as a consequence of convexity, any measure that obeys (P2s) also satisfies (P2) for the class of PI operations, in analogy to e.g.~\cite{baumgratz2014quantifying}.
As in Ref.~\cite{bischof2018resource} we can show that POVM-coherence measures, by construction, inherit the properties of the underlying block-coherence measure.

\begin{proposition}\label{prop:smon}
Let $C(\rho,\E)$ be a POVM-based coherence measure derived via~\eqref{pbcmeas} from a block-coherence measure $C(\rho',\P)$ that obeys strong monotonicity (B2s). Then, $C(\rho,\E)$ obeys strong monotonicity (P2s) with respect to PI operations.
\end{proposition}
\begin{proof} --
In the following, we make use of the constructions from Sec.~\ref{sec:pickraus}. Let $\{K_l\}$ be a set of POVM-incoherent Kraus operators, leading to the post-measurement states $\rho_l=K_l\rho K_l\ad/p_l$. Embedding these yields Naimark space operators given by 
\begin{align}
p_l\Ec[\rho_l]=TK_l\rho K_l\ad T\ad&=TT\ad K_l' T\rho T\ad (K_l')\ad TT\ad \ncl
&=\Pie K_l' \Ec[\rho] (K_l')\ad \Pie,
\end{align}
where we have used $\Ec[\rho]=T\rho T\ad$, Eq.~\eqref{pickraus} and $\Pie=TT\ad$.
Since $\Ec[\rho]=\Pie\Ec[\rho]\Pie$, we employ Eq.~\eqref{spkraus} twice to obtain the following simplification:
\begin{align}
\Pie K_l' \Ec[\rho] (K_l')\ad\Pie &= \Pie K_l' \Pie\Ec[\rho]\Pie (K_l')\ad\Pie \ncl
&= K_l'\Ec[\rho](K_l')\ad.
\end{align}
Thus, we have shown that $p_l\Ec[\rho_l] = K_l'\Ec[\rho](K_l')\ad$, which immediately implies the desired relation:
\begin{align}
\sum_lp_lC(\rho_l,\E) &= \sum_lp_lC(\Ec[\rho_l],\P) \ncl
&= \sum_lp_lC(K_l' \Ec[\rho] (K_l')\ad/p_l,\P) \ncl
&\leq C(\Ec[\rho],\P) = C(\rho,\E).
\end{align}
In the first and last line we have used Eq.~\eqref{pbcmeas} and the inequality holds since $C(\rho',\P)$ is by assumption strongly monotonic (B2s) with respect to block-incoherent Kraus operators $K_l'$.
\end{proof}

An example is given by the relative entropy of block coherence $C_{\rel}(\rho',\P)$, which satisfies (B2s), as one can prove analogously to Ref.~\cite{baumgratz2014quantifying} for the standard coherence measure. Thus, Prop.~\ref{prop:smon} implies that the POVM-coherence measure $C_{\rel}(\rho,\E)$ from Eq.~\eqref{repbc} is strongly monotonic.

\section{More measures of POVM-based coherence}
\label{sec:robcoh}

So far, the relative-entropy-based quantifier introduced in Ref.~\cite{bischof2018resource} is the only known well-defined measure of POVM-based coherence. In this section we introduce further POVM-coherence measures, which are generalizations of standard coherence measures known in the literature~\cite{streltsov2016quantum}. As before, $\E$ is a POVM on $\H$ and $\P$ any Naimark extension of it on the space $\H'$. We denote by $\Sc$ ($\Sc'$) the set of density matrices on $\H$ ($\H'$).

First, we discuss distance-based block-coherence quantifiers, which are defined as
\begin{align}\label{dbbcmeas}
C(\rho',\P)=\inf_{\sigma\in\Sc'} D(\rho',\Delta[\sigma]),
\end{align}
where $D\geq0$ is a distance such that $D(\rho,\sigma)=0\Leftrightarrow\rho=\sigma$ and $\Delta$ is the block-dephasing operation from Eq.~\eqref{bdeph}. The infimum runs over quantum states $\sigma\in\Sc'$. In Ref.~\cite{bischof2018resource} it was shown that a distance-based quantifier satisfies monotonicity (B2) (see~\ref{bcmeasp}) if $D$ is contractive under quantum operations, that is, $D(\Lambda[\rho],\Lambda[\sigma])\leq D(\rho,\sigma)$ holds for any channel $\Lambda$.

Distance-based POVM-coherence measures $C(\rho,\E)$ are derived from the measures $C(\rho',\P)$~\eqref{dbbcmeas} via Eq.~\eqref{pbcmeas}. We show below that this class of measures is independent of the choice of Naimark extension. Importantly, this implies that the POVM-coherence measure coincides for von Neumann measurements with the corresponding standard coherence measure~\cite{streltsov2016quantum}.

\begin{observation}\label{obs}
Let $C(\rho,\E)$ be a POVM-based coherence measure that is well-defined, i.e., it is invariant under the choice of Naimark extension $\P$ in Eq.~\eqref{pbcmeas}. Then, in the special case of orthogonal rank-1 (von Neumann) measurements, $C(\rho,\E)$ is equal to its counterpart in standard coherence theory.
\end{observation}
\begin{proof} --
The assertion holds because for the POVM $E_i=\proj{i}$, the Naimark extension can be chosen as $\P=\E$ and the embedding can be chosen trivial, $\Ec[\rho]=\rho$.
Thus, the independence property together with Eq.~\eqref{pbcmeas} guarantee that the POVM-based measure generalizes the standard measure. Note that the same argument holds for projective measurements, where $E_i=P_i$.
\end{proof}

\begin{proposition}\label{prop:measind}
Any distance-based POVM-coherence measure $C(\rho,\E)$ defined via Eqs.~\eqref{pbcmeas} and~\eqref{dbbcmeas} is invariant under the choice of Naimark extension if the distance is contractive.
\end{proposition}

\begin{proof} --
Let $\P$, $\hat\P$ be two Naimark extensions of the same POVM $\E$ such that $\rank\hat P_i\leq\rank P_i$. The corresponding block-dephasing operations are denoted $\Delta,\hat\Delta$. We need to show that $C(\Ec[\rho],\P)= C(\Ec[\rho],\hat\P)$. 
In the Appendix~\ref{app:rne} we show that there exists a channel (completely positive trace-preserving map) $\Nc$ which satisfies $\Nc\circ\Ec=\Ec$ and $\Nc\circ\Delta=\hat\Delta\circ\Nc$~\cite{bischof2018resource}.

Let $C(\rho',\P)=D(\rho',\Delta[\sigma^*])$ be a distance-based block coherence measure, where $\sigma^*$ denotes a state that achieves the minimum. Then, it holds that
\begin{align}
C(\Ec[\rho],\P) &= D(\Ec[\rho],\Delta[\sigma^*]) \ncl
&\geq D(\Nc\circ\Ec[\rho],\Nc\circ\Delta[\sigma^*]) \ncl
&= D(\Ec[\rho],\hat\Delta\circ\Nc[\sigma^*]) \ncl
&= D(\Ec[\rho],\hat\Delta[\hat\sigma]) \geq C(\Ec[\rho],\hat\P),
\end{align}
where we have defined $\hat\sigma\defeq\Nc[\sigma^*]$.
In the first inequality we have used the contractivity of $D$. The reverse inequality $C(\Ec[\rho],\P)\leq C(\Ec[\rho],\hat\P)$ follows from similar arguments but is more straightforward: the optimal state $\hat\Delta[\hat\sigma^*]$ on the smaller Naimark space can be embedded in the larger Naimark space and suitably rotated such that it is incoherent with respect to $\Delta$. This is achieved by the channel $\hat\Nc\defeq\Uc\ad\circ\Qc$ which satisfies $\hat\Nc\circ\Ec=\Ec$ and $\hat\Nc\circ\hat\Delta=\Delta\circ\hat\Nc$,
see App.~\ref{app:rne}.
\end{proof}

\paragraph*{Example:} Consider the distance measure $D_{\geo}(\rho,\sigma)=1-F^2(\rho,\sigma)$, where the fidelity $F(\rho,\sigma)=\tr\sqrt{\sqrt\rho\sigma\sqrt\rho}$ quantifies how close two quantum states $\rho,\sigma$ are. We define the \emph{geometric POVM-based coherence} $C_{\geo}(\rho,\E)$ via Eqs.~\eqref{pbcmeas} and~\eqref{dbbcmeas} for the distance $D_{\geo}$. The fidelity satisfies $F^2(\Lambda[\rho],\Lambda[\sigma])\geq F^2(\rho,\sigma)$ for any quantum operation $\Lambda$~\cite{nielsenchuang}, from which follows that $C_{\geo}(\rho,\E)$ obeys monotonicity (P2). Observation~\ref{obs} implies that this measure generalizes the standard geometric coherence~\cite{streltsov2015measuring}.

In the following, we introduce and study the robustness of POVM-based coherence which generalizes the measure from~\cite{napoli2016robustness}. This quantity is derived from the robustness of block coherence, which is equal to the robustness of asymmetry from Ref.~\cite{piani2016robustness} for the U(1) symmetry group $\{U(\theta)=e^{-i\theta\sum_kkP_k}\}$. Let $\P$ be a projective measurement and $\Delta$ the corresponding dephasing operator~\eqref{bdeph}. We define the robustness of block coherence of a quantum state $\rho$ as
\begin{align}
C_{\rob}(\rho,\P)&=\min_{\tau,\delta\in\Sc}\bigl\{s\geq0:\frac{\rho+s\tau}{1+s}=\Delta[\delta] \bigr\} \label{rob1}\\
&=\min_{\delta\in\Sc}\bigl\{s\geq0 : \rho\leq (1+s) \Delta[\delta]\bigr\}. \label{rob2}
\end{align}
In other words, $C_{\rob}(\rho,\P)$ is the minimal mixing weight $s$ required to make $\rho$ block-incoherent. It is clear that the measure satisfies faithfulness (B1). Moreover, the arguments from Ref.~\cite{piani2016robustness} imply that $C_{\rob}(\rho,\P)$ satisfies convexity (B3), and strong monotonicity (B2s) under selective block-incoherent operations. Interestingly, the robustness measure can be related to the \emph{maximum relative entropy of block coherence}, which we define as $C_{\max}(\rho,\P)=\min_{\delta\in \Sc}\{\lambda\geq0 : \rho\leq2^\lambda \Delta[\delta]\}$~\cite{bu2017maximum}. By comparison with Eq.~\eqref{rob2} we infer that $C_{\max}(\rho,\P)=\log_2[1+C_{\rob}(\rho,\E)]$. A further characterization of $C_{\rob}$ is given in the Appendix~\ref{app:rob}.

Now, let $\E$ be a POVM and $\P$ any Naimark extension of it. We employ the standard construction from Eq.~\eqref{pbcmeas} to define the \emph{robustness of POVM-based coherence} as
\begin{align}\label{robpbc}
C_{\rob}(\rho,\E)\defeq C_{\rob}(\Ec[\rho],\P).
\end{align}
The following result establishes $C_{\rob}(\rho,\E)$ as a proper measure of POVM-coherence.

\begin{proposition}\label{prop:rob}
The robustness of POVM-based coherence $C_{\rob}(\rho,\E)$ is well-defined and a POVM-coherence measure that satisfies strong monotonicity (P2s). It admits the following form:
\begin{align}\label{robpbc}
C_{\rob}(\rho,\E)=\min_{\tau\in \Sc'}
\bigl\{s\geq0: s\tau_{i,j}=-A_i\rho A_j\ad\ \forall i\neq j \bigr\},
\end{align}
where $\tau=\sum_{i,j}\tau_{i,j}\ot\ketbra{i}{j}$ and $A_i=\sqrt{E_i}$. 
\end{proposition}

Observation~\ref{obs} implies that in the special case of von Neumann measurements $\E=\{\proj{i}\}$, $C_{\rob}(\rho,\E)$ coincides with the standard robustness of coherence~\cite{napoli2016robustness}. The evaluation of $C_{\rob}$ in Eq.~\eqref{robpbc} is a semidefinite program (SDP). It can be simplified to the following form suited for numerical computation, for example, via the open-source MATLAB-based toolbox YALMIP~\cite{Lofberg2004}: 
\begin{align}
&C_{\rob}(\rho,\E) 
=\min \sum_i\tr(\sigma_{i,i}) \ncl
&\textrm{s.t.}\quad \sigma_{i\neq j,j}=-A_i\rho A_j\ad,\quad \sum_{i,j}\sigma_{i,j}\ot\ketbra{i}{j}\geq0.
\end{align}
This form is obtained from Prop.~\ref{prop:rob} by setting $\sigma=s\tau$.

\begin{proof}[Proof of Prop.~\ref{prop:rob}] --
First, we prove that the definition of $C_{\rob}(\rho,\E)$ is not ambiguous as it leads to the same quantity for any Naimark extension $\P$ of $\E$.
Let $\P$, $\hat\P$ be two Naimark extensions of the same POVM $\E$ such that $\rank\hat P_i\leq\rank P_i$. The corresponding block-dephasing operations are denoted $\hat\Delta,\Delta$.
It is clear that $C_{\rob}(\Ec[\rho],\P)\leq C_{\rob}(\Ec[\rho],\hat\P)$ since the optimal state $\hat\Delta[\hat\delta^*]$ in Eq.~\eqref{rob2} on the smaller Naimark space can be embedded in the larger Naimark space and suitably rotated such that it is incoherent with respect to $\Delta$.
We proceed to prove the reverse inequality by employing the channel $\Nc$ from the proof of Prop.~\ref{prop:measind}. Take Eq.~\eqref{rob2} with optimal quantities 
$s^*,\delta^*$ and apply $\Nc$ to both sides of the constraint
\begin{align}
\Ec[\rho]\leq (1+s^*) \Delta[\delta^*]
\Rightarrow \Nc\circ\Ec[\rho]&\leq (1+s^*) \Nc\circ\Delta[\delta^*] \ncl
\Leftrightarrow\quad \Ec[\rho]&\leq (1+s^*) \hat\Delta[\hat\delta],
\end{align}
where we have defined $\hat\delta=\Nc[\delta^*]$. Thus, $C_{\rob}(\Ec[\rho],\hat\P) \leq s^* =C_{\rob}(\Ec[\rho],\P)$. Altogether, we conclude that $C_{\rob}(\rho,\E)$ is independent of the Naimark extension choice. Moreover, $C_{\rob}$ satisfies strong monotonicity (P2s) because of Prop.~\ref{prop:smon} and Property~2 in Ref.~\cite{piani2016robustness}.

In order to prove Eq.~\eqref{robpbc}, we use the following result established as Prop.~4 in Ref.~\cite{bischof2018resource}. Any POVM-coherence measure can be written as
\begin{align}\label{altpbcmeas}
C(\rho,\E) = C(\Ec_V[\rho],\{\1\ot\proj{i}\}),
\end{align}
with the embedding $\Ec_V[\rho] =V\rho\ot\proj{1}V\ad=\sum_{i,j}A_i\rho A_j\ad\ot\ketbra{i}{j}$ containing an interaction isometry $V$, and the Naimark extension $\{\1\ot\proj{i}\}$. By using that in this formulation, $\delta\in\Ic\Leftrightarrow\delta=\sum_i\delta_i\ot\proj{i}$ and employing the parameterization $\tau=\sum_{i,j}\tau_{i,j}\ot\ketbra{i}{j}$, we obtain
\medmuskip=2mu
\begin{align}
&\quad C_{\rob}(\rho,\E) \ncl
&=\!\!\min_{\tau,\delta\in\Sc'}\!\bigl\{s\geq0 :\sum_{i,j}(A_i\rho A_j\ad+s\tau_{i,j})\ot\ketbra{i}{j}=(1+s)\sum_i\delta_i\ot\proj{i}\bigr\} \ncl
&=\min_{\tau\in\Sc'}
\bigl\{s\geq0: s\tau_{i,j}=-A_i\rho A_j\ad\ \forall i\neq j \bigr\},
\end{align}
Note that the constraint for $i=j$ was neglected in the last line, since for any $s$ and state $\tau$ satisfying the last line, we can define $\delta_i = (A_i\rho A_i\ad+s\tau_{i,i})/(1+s)$, which directly implies that $\delta\geq0$ and $\tr\delta=1$.
\end{proof}

We also define the following quantifier, the \emph{$\ell_1$-norm of POVM-based coherence}: $C_{\ell_1}(\rho,\E)=\sum_{i\neq j}\trn{P_i\Ec[\rho] P_j}$, where $\trn{X}=\tr(\sqrt{X\ad X})$ denotes the trace norm. By making use of Eq.~\eqref{altpbcmeas} and that $\trn{X\ot Y}=\trn{X}\trn{Y}$ holds for operators $X,Y$, it is straightforward to show that a simplified, local expression holds
\begin{align}\label{l1normpbc}
C_{\ell_1}(\rho,\E)=\sum_{i\neq j}\trn{A_i\rho A_j\ad}.
\end{align}
This generalized coherence quantifier satisfies faithfulness (P1), see Prop.~5 in~\cite{bischof2018resource}, and convexity (P3). Since for a von Neumann measurement $C_{\ell_1}(\rho,\E)$ reduces to the standard $\ell_1$-norm of coherence, we can infer that the measure does not satisfy monotonicity (P2) for the class MPI in general, see Ref.~\cite{bu2016note}. However, $C_{\ell_1}$ satisfies (P2) under MPI for any two-outcome POVM $\E=\{E_i\}_{i=1}^2$, which follows from Proposition 9 of Ref.~\cite{aberg2006quantifying} together with Prop.~\ref{prop:smon}. We leave open for future work whether $C_{\ell_1}(\rho,\E)$ satisfies strong monotonicity (P2s) under PI, which holds for von Neumann measurements~\cite{baumgratz2014quantifying}.

\begin{figure*}[!ht]
\centering
\includegraphics[width=0.35\linewidth,valign=T]{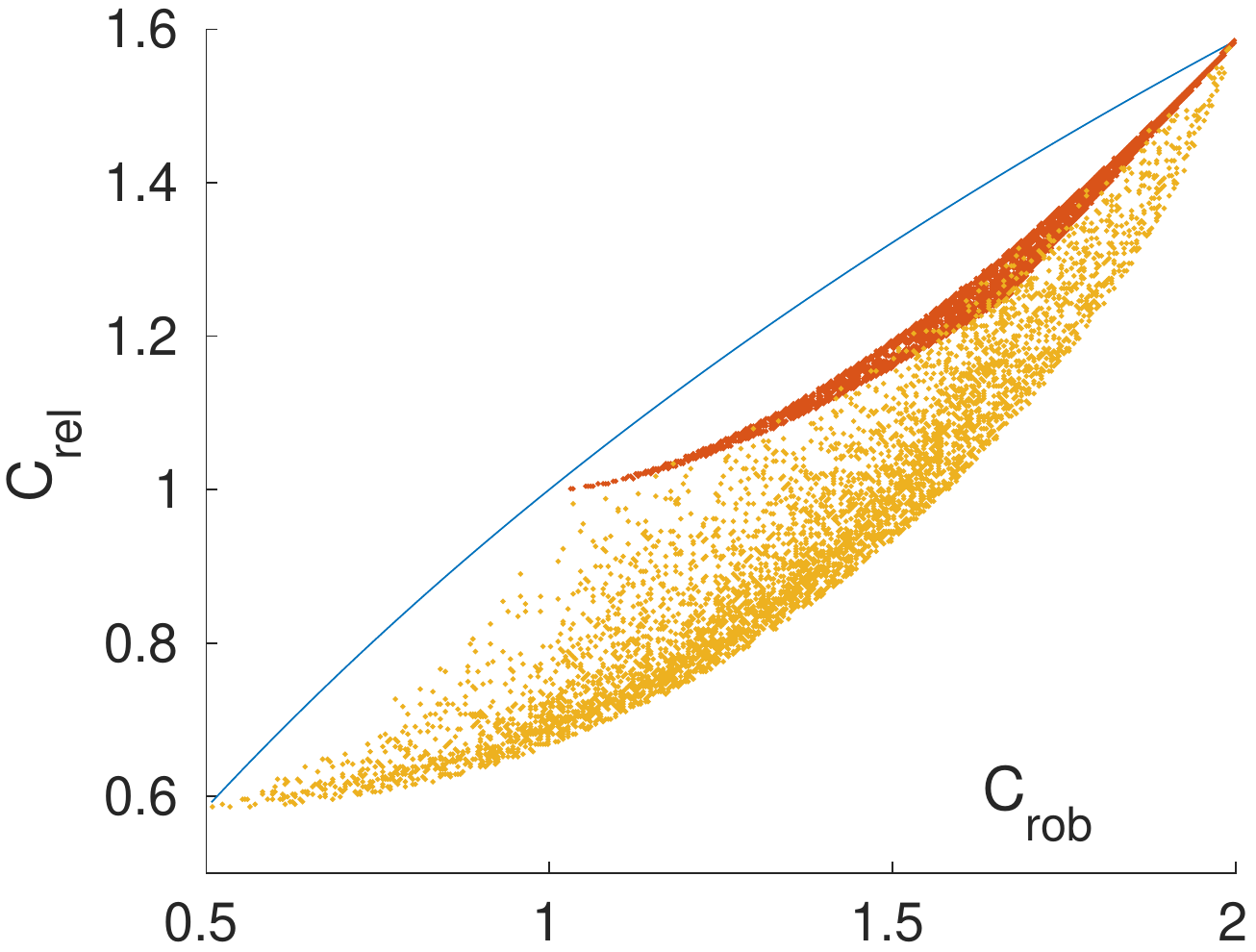}\hspace*{5em}
\includegraphics[width=0.35\linewidth,valign=T]{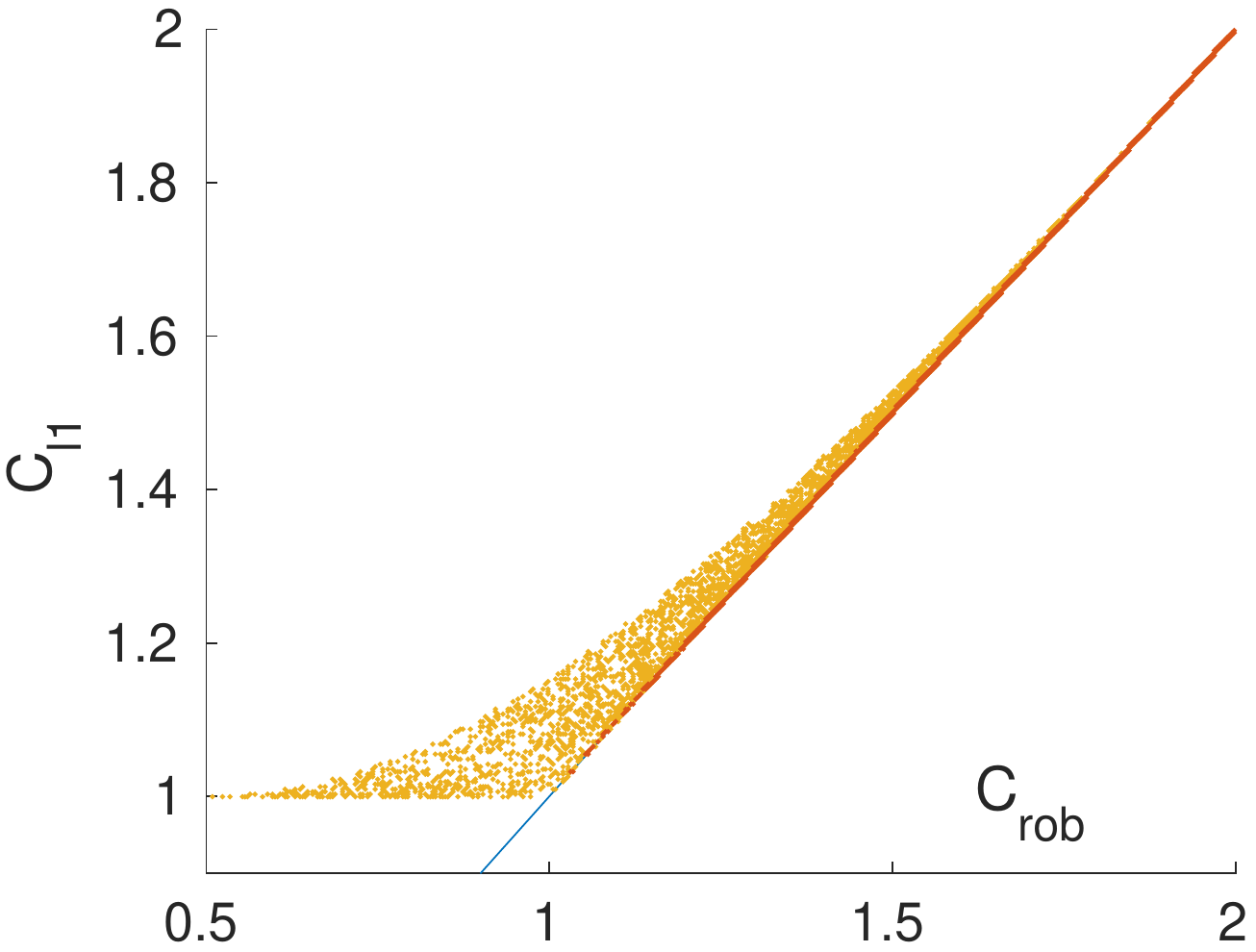}
\caption{\label{fig:robrel}
POVM-coherence measures in relation to the generalized robustness of coherence 
$s\defeq C_{\rob}(\rho,\E)$ for the qubit trine POVM $\E(\delta=1)$~\eqref{aux5}.
\emph{Left:} the blue line indicates the bound $C_{\rel}(\rho,\E)\leq \log_2(1+s)$ from Eq.~\eqref{robrel}. Red (yellow) dots represent randomly sampled pure (mixed) states. Similar to standard coherence theory~\cite{rana2017logarithmic}, the upper bound is not tight.
\emph{Right:} the blue, straight line indicates the graph of $C_{\ell_1}(\rho,\E)=s$, on which all pure states lie (red dots). The yellow dots represent mixed states for which $C_{\ell_1}(\rho,\E)\geq s$ holds~\eqref{robl1}.
}
\end{figure*}

For completeness, we show that $C_{\ell_1}(\rho,\E)$ is invariant under the choice of Naimark extension and unambigiously given by Eq.~\eqref{l1normpbc}. Given two Naimark extensions $\P$, $\hat\P$, we utilize the isometry $Q$ from App.~\ref{app:rne} satisfying $P_iQ = Q\hat P_i$. Further, we employ the unitary $U$ on the larger Naimark space with properties $UP_i=P_iU$ and $U\Pie=Q\Pie$, where $\Pie$ is the projector onto the embedded original space $\He$. Since the trace norm is invariant under multiplication by isometries $V,W$, $\trn{X} = \trn{VXW\ad}$, we have
\begin{align}
\trn{\hat P_i\Ec[\rho] \hat P_j} &= \trn{U\ad Q\hat P_i\Ec[\rho] \hat P_jQ\ad U} \ncl
&= \trn{P_iU\ad Q \Ec[\rho]Q\ad U P_j} = \trn{P_i\Ec[\rho]P_j}.\nonumber\qed
\end{align}
  
The following result establishes general relations between POVM-coherence measures that are visualized in Fig.~\ref{fig:robrel}. These findings generalize results from Ref.~\cite{rana2017logarithmic}.

\begin{proposition}
Given an $n$-outcome POVM $\E$, the following inequalities hold for the measures from Eqs.~\eqref{repbc},~\eqref{robpbc},~\eqref{l1normpbc}:
\begin{align}
&C_{\rob}(\rho,\E)\leq C_{\ell_1}(\rho,\E) \leq n-1, \label{robl1} \\
&C_{\rel}(\rho,\E) \leq \log_2[1+C_{\rob}(\rho,\E)]. \label{robrel}
\end{align}
Moreover, $C_{\rob}(\psi,\E)=C_{\ell_1}(\psi,\E)$ holds for any pure state $\psi$.
\end{proposition}

\begin{proof} --
First, we prove $C_{\rob}(\rho,\E) \leq n-1$ by showing that $C_{\rob}(\rho',\P) \leq n-1$ for any $n$-outcome projective measurement $\P$ and any state $\rho'\in\Sc'$.
For that, define $K_{i,j}=(P_i-P_j)/\sqrt2$ and consider the expression
\begin{align}
\sum_{i,j} K_{i,j}\rho K_{i,j}\ad 
&= \frac12 \sum_{i,j}(P_i-P_j)\rho(P_i-P_j) \ncl
&= \sum_{i,j}P_i \rho P_i - \sum_{i,j}P_i \rho P_j \\ 
&= n\sum_iP_i \rho P_i - \sum_{i,j}P_i \rho P_j= (n\Delta-\id)[\rho]. \nonumber
\end{align}
Consequently, the map $(n\Delta-\id)$ admits a Kraus decomposition and is thus completely positive. This implies that $n\Delta[\rho']-\rho'\geq0$ holds for any quantum state $\rho'$. Hence, we obtained $\rho'\leq n\Delta[\rho']$ and by comparison with Eq.~\eqref{rob2} we conclude that $C_{\rob}(\rho',\P)=s\leq n-1$.

The relation $C_{\ell_1}(\rho,\E) \leq n-1$ can be shown by evaluating the underlying block-coherence measure for a maximally coherent state. The latter is given by $\ket{\Psi_{\operatorname{m}}}=\frac1{\sqrt n}\sum_i\ket{\varphi_i}$ with pure block-incoherent states $\ket{\varphi_i}$ defined in Eq.~\eqref{purebic}. This leads to
\begin{align}
C_{\ell_1}(\ket{\Psi_{\operatorname{m}}},\P) 
&= \frac1n\sum_{i\neq j}\trn{\sum_{k,l} P_i \ketbra{\varphi_k}{\varphi_l}P_j}
= \frac1n\sum_{i\neq j}\trn{\,\ketbra{\varphi_i}{\varphi_j}\,} \ncl 
&= \frac1n\sum_{i\neq j}1 = \frac1nn(n-1) = n-1.
\end{align}

In the Appendix~\ref{app:rob} we show a further SDP characterization of the robustness of POVM-based coherence. Moreover, this form is used to show that $C_{\rob}(\psi,\E)= C_{\ell_1}(\psi,\E)$ for pure states and $C_{\rob}(\rho,\E)\leq C_{\ell_1}(\rho,\E)$ in general.

Finally, we show Eq.~\eqref{robrel} similar to Ref.~\cite{rana2017logarithmic}. Let $s^*,\delta^*$ be the the optimal quantities for $C_{\rob}(\rho,\E)=C_{\rob}(\Ec[\rho],\P)$ in Eq.~\eqref{rob2}. Using the abbreviation $\rho_\Ect=\Ec[\rho]$, it holds that $C_{\rel}(\rho_\Ect,\P)=S(\rho_\Ect || \Delta[\rho_\Ect])\leq S(\rho_\Ect || \Delta[\delta^*])$. Moreover
\begin{align}\label{aux11}
&S(\rho_\Ect || \Delta[\delta^*]) \ncl
=&\tr\Bigl[\rho_\Ect\bigl(\log_2\rho_\Ect-\log_2\frac{(1+s^*)\Delta[\delta^*]}{(1+s^*)}\bigr)\Bigr] \\
=&\log_2(1+s^*) + \tr[\rho_\Ect(\log_2\rho_\Ect-\log_2(1+s^*)\Delta[\delta^*])], \nonumber
\end{align}
where we have used the definition of the relative entropy $S(\rho\rvert\lvert\sigma)=\tr[\rho(\log_2\rho-\log_2\sigma)]$.
On the other hand, Eq.~\eqref{rob2} implies that $\rho_\Ect\leq(1+s^*)\Delta[\delta^*]$. The latter relation together with the fact that the logarithm is operator-monotone yields that the second term in~\eqref{aux11} (last line) is non-positive. We conclude that $C_{\rel}(\rho,\E)\leq S(\rho_\Ect || \Delta[\delta^*])\leq\log_2(1+s^*)$ implying the desired relation.  
\end{proof}

\section{Conclusion and Outlook}
\label{sec:conclusion}

We presented several results on the resource-theoretical concept of coherence with respect to a general quantum measurement. We expect these advances to clarify the role of quantum coherence in information technologies employing nonprojective measurements. 
In particular, we discussed selected features of POVM-based coherence theory that are distinct from the standard resource theory of coherence. Moreover, we established a probabilistic framework of free transformations in conjunction with resource measures. 
This led to the introduction of new, strongly monotonic POVM-based coherence measures 
that generalize well-known coherence measures. We also established relations among the new measures. Finally, we showed that the relative-entropy-based resource measure is equal to the cryptographic randomness gain, providing an important operational meaning to the concept of coherence with respect to a measurement.

Together with Ref.~\cite{bischof2018resource}, we have paved the way for a detailed operational analysis of POVM-based coherence as a resource, akin to what has been achieved in the standard resource theory of coherence~\cite{winter2016operational,chitambar2016comparison,chitambar2016critical,
yadin2016quantum}. The operational analysis includes the investigation of resource distillation and dilution in the asymptotic and single-shot regime, see~\cite{zhao2018one,regula2018one,lami2019completing}. In particular, it is open whether our theory is reversible, or there are bound resources for a given class of POVM-incoherent operations~\cite{zhao2019one,lami2019generic}. An important step towards this goal would consist in a possible simplification of our constructions, e.g., of the MPI and PI operations. 
Moreover, we expect that virtually all known coherence measures and channel classes~\cite{streltsov2016quantum} can be generalized to POVMs. It is likely that more operational interpretations of POVM-based coherence measures can be found which link the resource theory to interesting applications in quantum information science. Finally, future work should address the connection of POVM-based coherence with other notions of nonclassicality such as entanglement and purity~\cite{streltsov2015measuring,streltsov2018maximal}.

\begin{acknowledgments}
We acknowledge financial support from the German Federal Ministry of Education and Research (BMBF).
F.B.\ gratefully acknowledges support from Evangelisches Studienwerk Villigst and from Strategischer Forschungsfonds of the Heinrich Heine University D\"usseldorf. 
\end{acknowledgments}

\noindent
\section*{Appendix}\label{app}
\appendix

\section{Relating Naimark extensions of a POVM}
\label{app:rne}

In the Supplemental Material of Ref.~\cite{bischof2018resource} several relations between Naimark extensions of a POVM were established. In this section, we provide an overview of these results which are used to show that the constituents of our POVM-based coherence theory do not depend on the choice of Naimark extension. In particular, we prove Prop.~\ref{prop:pickraus} at the end of this section.

Let $\P$, $\hat\P$ be two Naimark extensions of the same $n$-outcome POVM $\E$ such that $\rank \hat P_i\leq\rank P_i$. There exists an isometry $Q\colon\hat\H\to\H'$ from the smaller Naimark spacer to the larger Naimark space such that 
\begin{align}
&P_iQ = Q\hat P_i \quad\textrm{and} \label{a1} \\
&\Qc\circ\hat\Delta=\Delta\circ\Qc,
\end{align}
where we have defined the isometric channel $\Qc[X]=QXQ\ad$ and $\hat\Delta[X]=\sum_i\hat P_i X \hat P_i$ denotes the block-dephasing operator.

Moreover, it was shown that there exists a unitary $U$ on the larger Naimark space such that~\cite{bischof2018resource}
\begin{align}
&Q\Pie = U\Pie \quad\textrm{and} \label{a2} \\
&\Qc\circ\Ec = \Uc\circ\Ec,
\end{align}
where $\Ec[X]=TXT\ad=X\op0$ denotes the embedding operation, see Sec.~\ref{sec:pbcoh}. This unitary can be chosen to be block-diagonal such that it commutes with the Naimark extension effects
\begin{align}
&UP_i=P_iU \quad\textrm{and} \label{a3} \\
&\Delta\circ\Uc = \Uc\circ\Delta.
\end{align}

The channel $\Qc\ad[\rho]=Q\ad\rho Q$ is completely positive but not trace-preserving in general. Define the projector $S\defeq QQ\ad$ and its complement $S^\perp = \1-S$ for which holds that $S^\perp Q=0$.
We define the completely positive map 
\begin{align}
\Tc[\rho]\defeq\tr(S^\perp\rho)\1/d_{\min},
\end{align}
which has Kraus operators
\begin{align}\label{lab}
L_{\hat a,b}=\frac1{\sqrt{ d_{\min}}}\ketbra{\hat a}{b}S^\perp, 
\end{align}
where $\{\ket{\hat a}\}$ ($\{\ket{b}\}$) denotes an orthonormal basis of the smaller (larger) Naimark space.
We choose as output basis $\ket{\hat a}\in\hat\H$ an incoherent basis with respect to $\hat P_i$. Consequently, $L_{\hat a,b}$ cannot create coherence for any input.
Define the operators
\begin{align}\label{rm}
R_{m}=
\begin{cases}
Q\ad & \textrm{for } m=0 \\
L_{\hat a,b} & \textrm{for } m\geq1,
\end{cases} 
\end{align}
where the index $m$ for $m\geq1$ runs over all combinations of $(\hat a,b)$.
The set $\{R_m\}$ is a set of Kraus operators for the channel 
\begin{align}
\Rc=\Qc\ad+\Tc.
\end{align}
It holds that $\Rc\circ\Qc=\id$, i.e., $\Rc$ is a reversal channel of the isometric channel $\Qc$. One can show that the following equation holds~\cite{bischof2018resource}
\begin{align}
\hat\Delta\circ\Rc = \Rc\circ\Delta.
\end{align}
In addition, it holds that $\Tc\circ\Ec[\rho]=\tr(S^\perp\Ec[\rho])\1/d_{\min}=0$ and therefore
\begin{align}
\Rc\circ\Ec = \Qc\ad\circ\Ec.
\end{align}

Finally, we define the following channel from operators on the larger Naimark space to operators on the smaller Naimark space: 
\begin{align}
\Nc\defeq\Rc\circ\Uc,\quad \textrm{which satisfies} \\
\Nc\circ\Ec=\Ec\quad\textrm{and}\quad \Nc\circ\Delta=\hat\Delta\circ\Nc.
\end{align}
The first equality follows from $\Nc\circ\Ec=\Rc\circ\Uc\circ\Ec=\Rc\circ\Qc\circ\Ec=\Ec$. The second equality follows from $\Nc\circ\Delta=\Rc\circ\Uc\circ\Delta=\hat\Delta\circ\Rc\circ\Uc= \hat\Delta\circ\Nc$.

\subsection*{Proof of Proposition~\ref{prop:pickraus}}

\noindent\textbf{Proposition 2.}\ \ 
The set containing all POVM-incoherent (PI) Kraus operators $K_l$ does not depend on the choice of Naimark extension used to define it, see Eq.~\eqref{pickraus}.

\begin{proof} --
Let $\P$, $\hat\P$ be two Naimark extensions of the same POVM $\E$ such that $\rank \hat P_i\leq\rank P_i$. Let $\{K_l=T\ad K_l' T\}$ be the set of POVM-incoherent Kraus operators defined via incoherent operators $\{K_l'\}$ of the ``larger'' Naimark extension $\P$, see Eq.~\eqref{pickraus}. Consider the MBI channel $\Gamma[\rho']=\sum_{l}K_l'\rho' (K_l')\ad$ on the larger Naimark space. The channel $\hat\Gamma\defeq\Rc\circ\Uc\circ\Gamma\circ\Uc\ad\circ\Qc$ is a MBI channel on the smaller Naimark space, that leads to the same (local) MPI operation $\Lambda_\mpi$~\cite{bischof2018resource}. We consider the following Kraus decomposition of the channel:
\begin{align}
\hat\Gamma[\hat\rho]
&=\sum_{m,l}R_mUK'_lU\ad Q\hat\rho Q\ad U(K'_l)\ad U\ad R_m\ad \ncl
&=\sum_{m,l}\hat K_{m,l}\hat\rho \hat K_{m,l}\ad, \ncl 
\hat K_{m,l} &\defeq R_mUK'_lU\ad Q, \label{kml}
\end{align}
where $R_m$ was defined in Eq.~\eqref{rm}. 

We proceed to show that the set $\{\hat K_{m,l}\}$ 
\begin{enumerate}[nosep]
\item[i)] satisfies $\sum_{m,l}\hat K_{m,l}\ad\hat K_{m,l}=1$,
\item[ii)] has the property that each element is incoherent w.r.t. $\hat\P$,
\item[iii)] leads to the previous set of PI Kraus operators, more precisely, $T\ad \hat K_{m,l} T = \delta_{m,0}K_l$.
\end{enumerate}

The first claim holds since $\{\hat K_{m,l}\}$ is a set of Kraus operators of $\hat\Gamma$, which is a completely positive trace-preserving map~\cite{bischof2018resource}.

For the second claim, consider a block-incoherent pure state $\ket{\varphi_i}=\hat P_i \ket{\varphi_i}$, for which holds:
\begin{align}\label{a4}
\hat K_{m,l}\ket{\varphi_i} &= \hat K_{m,l} \hat P_i \ket{\varphi_i} \ncl
&= R_mUK'_l P_i U\ad Q \ket{\varphi_i} \ncl
&= R_mUP_{f(i)}K'_l P_i U\ad Q \ket{\varphi_i} \ncl
&= R_mP_{f(i)}UK'_l P_i U\ad Q \ket{\varphi_i} \ncl
&= \begin{cases}
\hat P_{f(i)}Q\ad UK'_l P_i U\ad Q \ket{\varphi_i} & \textrm{for } m=0 \\
L_{\hat a,b}P_{j(i)}UK'_l P_i U\ad Q \ket{\varphi_i} & \textrm{else.}
\end{cases}
\end{align}
The second equation makes use of~\eqref{a1} and~\eqref{a3}. In the third line we have used that for an incoherent input, the output of $K_l'$ is incoherent~\eqref{aux10}. Finally, the last equation follows from the definition of $R_m$~\eqref{rm}.
Note that in any case, the output of the Kraus operator in~\eqref{a4} is incoherent, see~\eqref{lab}.

For the third claim, we evaluate:
\begin{align}
T\ad \hat K_{m,l} T &= T\ad R_mUK'_lU\ad Q T \ncl
&= T\ad R_mUK'_l T \ncl
&= T\ad R_mU\Pie K'_l T \ncl
&= T\ad R_mQ\Pie K'_l T \ncl
&= \delta_{m,0} T\ad K'_lT = \delta_{m,0}K_l.
\end{align}
In the first line, the definition of $\hat K_{m,l}$~\eqref{kml} was inserted.
The second and fourth line utilize the relations $UT=QT$ and $\Pie=TT\ad$. In the third line, we have used that $K_l'$ is subspace-preserving~\eqref{spkraus}.
Finally, for the last line, note that according to~\eqref{rm}, $R_0Q=Q\ad Q=\1$, and $R_mQ=0$ for $m\geq1$.
\end{proof}

\section{Alternative SDP for generalized robustness measure}
\label{app:rob}

In Ref.~\cite{piani2016robustness} it was shown that the robustness of block-coherence (asymmetry) can be expressed by the following SDP:
\begin{align}
&C_{\rob}(\rho,\P) = \max\tr(X\rho)-1, \ncl
&\textrm{s.t.}\quad X\geq0,\quad \Delta[X]=\1.
\end{align}
where $\Delta[X]=\sum_iP_iXP_i$ denotes the block-dephasing operation. Consider the POVM-coherence measure $C_{\rob}(\rho,\E)=C_{\rob}(\Ec_V[\rho],\P)$, where $P_i=\1\ot\proj{i}$ and $\Ec_V[\rho]= \sum_{i,j}A_i\rho A_j\ad\ot\ketbra{i}{j}$, see Eq.~\eqref{altpbcmeas}. If we write $X=\sum_{i,j}X_{i,j}\ot\ketbra{i}{j}$, we directly obtain the SDP: 
\begin{align}
&C_{\rob}(\rho,\E)=\max\tr \bigl(\sum_{i,j}X_{j,i}A_i\rho A_j\ad\bigr)-1 \ncl
&\textrm{s.t.}\quad \sum_{i,j}X_{i,j}\ot\ketbra{i}{j}\geq 0,\quad X_{i,i}=\1.
\end{align}

Employing this form, we are able to show that $C_{\rob}(\rho,\E)\leq C_{\ell_1}(\rho,\E)$ as follows:
\begin{align}\label{robrelproof}
C_{\rob}(\rho,\E) &= \max_{X\geq0,X_{i,i}=\1} \sum_{i,j}\tr(X_{j,i}A_i\rho A_j\ad)-1 \ncl
&= \max_{X\geq0,X_{i,i}=\1} \sum_{i\neq j}\tr(X_{j,i}A_i\rho A_j\ad) \ncl
&= \max_{X\geq0,X_{i,i}=\1} 2\sum_{i<j}\operatorname{Re}\tr(X_{j,i}A_i\rho A_j\ad) \ncl
&\leq \max_{X\geq0,X_{i,i}=\1} 2\sum_{i<j}\lvert\tr(X_{j,i}A_i\rho A_j\ad)\rvert \ncl
&\leq 2\sum_{i<j}\max_{\lVert X_{i,j}\rVert_\infty\leq1}\lvert\tr(X_{j,i}A_i\rho A_j\ad)\rvert \ncl
&= 2\sum_{i<j}\trn{A_i\rho A_j\ad} = C_{\ell_1}(\rho,\E).
\end{align}
For the second inequality, we have used that $X\geq0, X_{i,i}=1$ implies $\lvert\lvert X_{i,j}\rvert\rvert_\infty\leq\1$, where $\lvert\lvert X \rvert\rvert_\infty$ denotes the largest singular value of $X$. Then, we employed the variational characterization of the trace norm, $\trn{R}=\max_{\lvert\lvert L\rvert\rvert_\infty\leq1}\modu{\tr(L\ad R)}$, which follows from the duality property of the Schatten norms~\cite{watrous2018theory}.

We proceed that show that $C_{\rob}(\psi,\E) = C_{\ell_1}(\psi,\E)$ holds for any pure state $\psi\defeq\proj{\psi}$. For indices $i,j$, consider the rank one operator $A_i\proj{\psi} A_j\ad = \sqrt{p_ip_j}\ketbra{\phi_i}{\phi_j}$ with $p_i\defeq\bra{\psi}A_i\ad A_i\ket{\psi}\leq1$. The vectors $\ket{\phi_i}=\frac1{\sqrt{p_i}}A_i\ket{\psi}$ are normalized and not necessarily orthogonal. Evaluating $C_{\ell_1}(\psi,\E)$ yields
\begin{align}\label{l1pure}
C_{\ell_1}(\psi,\E) &= \sum_{i\neq j}\trn{A_i\proj{\psi} A_j\ad} \ncl
&= \sum_{i\neq j}\sqrt{p_ip_j}\trn{\,\ketbra{\phi_i}{\phi_j}\,} \ncl
&= \sum_{i\neq j}\sqrt{p_ip_j}.
\end{align}
We define the hermitian operator $\tilde X=\sum_{i,j}\tilde X_{i,j}\ot\ketbra{i}{j}$ as
\begin{align}\label{defx}
\tilde X=\sum_{i,j}\ketbra{\phi_i}{\phi_j}\ot\ketbra{i}{j}+\sum_i(\1-\proj{\phi_i})\ot\proj{i}.
\end{align}
It holds that $\tilde X\geq0$ since the first term can be written as $\proj{\Omega}\geq0$ with $\ket{\Omega}=\sum_i\ket{\phi_i}\ot\ket{i}$, while the second term is in spectral decomposition form and apparently positive semidefinite. Moreover, the diagonal blocks of $\tilde X$ are equal to the identity, $\tilde X_{i,i}=\1$. Thus, $\tilde X$ is element of the feasible set of operators $X$ used to obtain $C_{\rob}(\psi,\E) = \max_{X\geq0,X_{i,i}=\1} \sum_{i\neq j}\tr(X_{j,i}A_i\proj{\psi} A_j\ad)$. Hence, it follows that
\begin{align}\label{robpure}
C_{\rob}(\psi,\E) &\geq \sum_{i\neq j}\tr(\tilde X_{j,i}A_i\proj{\psi} A_j\ad) \ncl
&= \sum_{i\neq j}\sqrt{p_ip_j}\tr(\ketbra{\phi_j}{\phi_i}\, \ketbra{\phi_i}{\phi_j}) \ncl
&= \sum_{i\neq j}\sqrt{p_ip_j}.
\end{align}
By comparing~\eqref{l1pure} and~\eqref{robpure}, we infer that $C_{\rob}(\psi,\E)\geq C_{\ell_1}(\psi,\E)$ holds for any pure state $\psi$. Combining this with the inequality $C_{\rob}(\rho,\E)\leq C_{\ell_1}(\rho,\E)$ for general states $\rho$, we conclude that $C_{\rob}=C_{\ell_1}$ holds for pure states and any POVM. \qed

\bibliography{C:/Users/FB/Desktop/Doktorarbeit/bibqi}

\begin{thebibliography}{50}
\expandafter\ifx\csname natexlab\endcsname\relax\def\natexlab#1{#1}\fi
\expandafter\ifx\csname bibnamefont\endcsname\relax
  \def\bibnamefont#1{#1}\fi
\expandafter\ifx\csname bibfnamefont\endcsname\relax
  \def\bibfnamefont#1{#1}\fi
\expandafter\ifx\csname citenamefont\endcsname\relax
  \def\citenamefont#1{#1}\fi
\expandafter\ifx\csname url\endcsname\relax
  \def\url#1{\texttt{#1}}\fi
\expandafter\ifx\csname urlprefix\endcsname\relax\def\urlprefix{URL }\fi
\providecommand{\bibinfo}[2]{#2}
\providecommand{\eprint}[2][]{\url{#2}}

\bibitem[{\citenamefont{Ac{\'\i}n et~al.}(2007)\citenamefont{Ac{\'\i}n,
  Brunner, Gisin, Massar, Pironio, and Scarani}}]{acin2007device}
\bibinfo{author}{\bibfnamefont{A.}~\bibnamefont{Ac{\'\i}n}},
  \bibinfo{author}{\bibfnamefont{N.}~\bibnamefont{Brunner}},
  \bibinfo{author}{\bibfnamefont{N.}~\bibnamefont{Gisin}},
  \bibinfo{author}{\bibfnamefont{S.}~\bibnamefont{Massar}},
  \bibinfo{author}{\bibfnamefont{S.}~\bibnamefont{Pironio}}, \bibnamefont{and}
  \bibinfo{author}{\bibfnamefont{V.}~\bibnamefont{Scarani}},
  \bibinfo{journal}{Physical Review Letters} \textbf{\bibinfo{volume}{98}},
  \bibinfo{pages}{230501} (\bibinfo{year}{2007}).

\bibitem[{\citenamefont{Arnon-Friedman
  et~al.}(2018)\citenamefont{Arnon-Friedman, Dupuis, Fawzi, Renner, and
  Vidick}}]{arnon2018practical}
\bibinfo{author}{\bibfnamefont{R.}~\bibnamefont{Arnon-Friedman}},
  \bibinfo{author}{\bibfnamefont{F.}~\bibnamefont{Dupuis}},
  \bibinfo{author}{\bibfnamefont{O.}~\bibnamefont{Fawzi}},
  \bibinfo{author}{\bibfnamefont{R.}~\bibnamefont{Renner}}, \bibnamefont{and}
  \bibinfo{author}{\bibfnamefont{T.}~\bibnamefont{Vidick}},
  \bibinfo{journal}{Nature Communications} \textbf{\bibinfo{volume}{9}},
  \bibinfo{pages}{459} (\bibinfo{year}{2018}).

\bibitem[{\citenamefont{Brand{\~a}o and Gour}(2015)}]{brandao2015reversible}
\bibinfo{author}{\bibfnamefont{F.~G. S.~L.} \bibnamefont{Brand{\~a}o}}
  \bibnamefont{and} \bibinfo{author}{\bibfnamefont{G.}~\bibnamefont{Gour}},
  \bibinfo{journal}{Physical Review Letters} \textbf{\bibinfo{volume}{115}},
  \bibinfo{pages}{070503} (\bibinfo{year}{2015}).

\bibitem[{\citenamefont{Liu et~al.}(2017)\citenamefont{Liu, Hu, and
  Lloyd}}]{liu2017resource}
\bibinfo{author}{\bibfnamefont{Z.-W.} \bibnamefont{Liu}},
  \bibinfo{author}{\bibfnamefont{X.}~\bibnamefont{Hu}}, \bibnamefont{and}
  \bibinfo{author}{\bibfnamefont{S.}~\bibnamefont{Lloyd}},
  \bibinfo{journal}{Physical Review Letters} \textbf{\bibinfo{volume}{118}},
  \bibinfo{pages}{060502} (\bibinfo{year}{2017}).

\bibitem[{\citenamefont{Chitambar and Gour}(2019)}]{chitambar2018quantum}
\bibinfo{author}{\bibfnamefont{E.}~\bibnamefont{Chitambar}} \bibnamefont{and}
  \bibinfo{author}{\bibfnamefont{G.}~\bibnamefont{Gour}},
  \bibinfo{journal}{Reviews of Modern Physics} \textbf{\bibinfo{volume}{91}},
  \bibinfo{pages}{025001} (\bibinfo{year}{2019}).

\bibitem[{\citenamefont{Horodecki
  et~al.}(2003{\natexlab{a}})\citenamefont{Horodecki, Horodecki, Horodecki,
  Horodecki, Oppenheim, Sen, Sen et~al.}}]{horodecki2003local}
\bibinfo{author}{\bibfnamefont{M.}~\bibnamefont{Horodecki}},
  \bibinfo{author}{\bibfnamefont{K.}~\bibnamefont{Horodecki}},
  \bibinfo{author}{\bibfnamefont{P.}~\bibnamefont{Horodecki}},
  \bibinfo{author}{\bibfnamefont{R.}~\bibnamefont{Horodecki}},
  \bibinfo{author}{\bibfnamefont{J.}~\bibnamefont{Oppenheim}},
  \bibinfo{author}{\bibfnamefont{A.}~\bibnamefont{Sen}},
  \bibinfo{author}{\bibfnamefont{U.}~\bibnamefont{Sen}}, \bibnamefont{et~al.},
  \bibinfo{journal}{Physical Review Letters} \textbf{\bibinfo{volume}{90}},
  \bibinfo{pages}{100402} (\bibinfo{year}{2003}{\natexlab{a}}).

\bibitem[{\citenamefont{Horodecki et~al.}(2009)\citenamefont{Horodecki,
  Horodecki, Horodecki, and Horodecki}}]{horodecki2009quantum}
\bibinfo{author}{\bibfnamefont{R.}~\bibnamefont{Horodecki}},
  \bibinfo{author}{\bibfnamefont{P.}~\bibnamefont{Horodecki}},
  \bibinfo{author}{\bibfnamefont{M.}~\bibnamefont{Horodecki}},
  \bibnamefont{and}
  \bibinfo{author}{\bibfnamefont{K.}~\bibnamefont{Horodecki}},
  \bibinfo{journal}{Reviews of Modern Physics} \textbf{\bibinfo{volume}{81}},
  \bibinfo{pages}{865} (\bibinfo{year}{2009}).

\bibitem[{\citenamefont{Horodecki
  et~al.}(2003{\natexlab{b}})\citenamefont{Horodecki, Horodecki, and
  Oppenheim}}]{horodecki2003reversible}
\bibinfo{author}{\bibfnamefont{M.}~\bibnamefont{Horodecki}},
  \bibinfo{author}{\bibfnamefont{P.}~\bibnamefont{Horodecki}},
  \bibnamefont{and}
  \bibinfo{author}{\bibfnamefont{J.}~\bibnamefont{Oppenheim}},
  \bibinfo{journal}{Physical Review A} \textbf{\bibinfo{volume}{67}},
  \bibinfo{pages}{062104} (\bibinfo{year}{2003}{\natexlab{b}}).

\bibitem[{\citenamefont{Marvian and Spekkens}(2013)}]{marvian2013theory}
\bibinfo{author}{\bibfnamefont{I.}~\bibnamefont{Marvian}} \bibnamefont{and}
  \bibinfo{author}{\bibfnamefont{R.~W.} \bibnamefont{Spekkens}},
  \bibinfo{journal}{New Journal of Physics} \textbf{\bibinfo{volume}{15}},
  \bibinfo{pages}{033001} (\bibinfo{year}{2013}).

\bibitem[{\citenamefont{Marvian and Spekkens}(2014)}]{marvian2014extending}
\bibinfo{author}{\bibfnamefont{I.}~\bibnamefont{Marvian}} \bibnamefont{and}
  \bibinfo{author}{\bibfnamefont{R.~W.} \bibnamefont{Spekkens}},
  \bibinfo{journal}{Nature communications} \textbf{\bibinfo{volume}{5}},
  \bibinfo{pages}{3821} (\bibinfo{year}{2014}).

\bibitem[{\citenamefont{Brandao et~al.}(2013)\citenamefont{Brandao, Horodecki,
  Oppenheim, Renes, and Spekkens}}]{brandao2013resource}
\bibinfo{author}{\bibfnamefont{F.~G. S.~L.} \bibnamefont{Brandao}},
  \bibinfo{author}{\bibfnamefont{M.}~\bibnamefont{Horodecki}},
  \bibinfo{author}{\bibfnamefont{J.}~\bibnamefont{Oppenheim}},
  \bibinfo{author}{\bibfnamefont{J.~M.} \bibnamefont{Renes}}, \bibnamefont{and}
  \bibinfo{author}{\bibfnamefont{R.~W.} \bibnamefont{Spekkens}},
  \bibinfo{journal}{Physical Review Letters} \textbf{\bibinfo{volume}{111}},
  \bibinfo{pages}{250404} (\bibinfo{year}{2013}).

\bibitem[{\citenamefont{Baumgratz et~al.}(2014)\citenamefont{Baumgratz, Cramer,
  and Plenio}}]{baumgratz2014quantifying}
\bibinfo{author}{\bibfnamefont{T.}~\bibnamefont{Baumgratz}},
  \bibinfo{author}{\bibfnamefont{M.}~\bibnamefont{Cramer}}, \bibnamefont{and}
  \bibinfo{author}{\bibfnamefont{M.~B.} \bibnamefont{Plenio}},
  \bibinfo{journal}{Physical Review Letters} \textbf{\bibinfo{volume}{113}},
  \bibinfo{pages}{140401} (\bibinfo{year}{2014}).

\bibitem[{\citenamefont{Winter and Yang}(2016)}]{winter2016operational}
\bibinfo{author}{\bibfnamefont{A.}~\bibnamefont{Winter}} \bibnamefont{and}
  \bibinfo{author}{\bibfnamefont{D.}~\bibnamefont{Yang}},
  \bibinfo{journal}{Physical Review Letters} \textbf{\bibinfo{volume}{116}},
  \bibinfo{pages}{120404} (\bibinfo{year}{2016}).

\bibitem[{\citenamefont{Streltsov et~al.}(2017)\citenamefont{Streltsov, Adesso,
  and Plenio}}]{streltsov2016quantum}
\bibinfo{author}{\bibfnamefont{A.}~\bibnamefont{Streltsov}},
  \bibinfo{author}{\bibfnamefont{G.}~\bibnamefont{Adesso}}, \bibnamefont{and}
  \bibinfo{author}{\bibfnamefont{M.~B.} \bibnamefont{Plenio}},
  \bibinfo{journal}{Rev. Mod. Phys.} \textbf{\bibinfo{volume}{89}},
  \bibinfo{pages}{041003} (\bibinfo{year}{2017}).

\bibitem[{\citenamefont{Coecke et~al.}(2016)\citenamefont{Coecke, Fritz, and
  Spekkens}}]{coecke2016mathematical}
\bibinfo{author}{\bibfnamefont{B.}~\bibnamefont{Coecke}},
  \bibinfo{author}{\bibfnamefont{T.}~\bibnamefont{Fritz}}, \bibnamefont{and}
  \bibinfo{author}{\bibfnamefont{R.~W.} \bibnamefont{Spekkens}},
  \bibinfo{journal}{Information and Computation}
  \textbf{\bibinfo{volume}{250}}, \bibinfo{pages}{59} (\bibinfo{year}{2016}).

\bibitem[{\citenamefont{Horodecki and
  Oppenheim}(2013)}]{horodecki2013quantumness}
\bibinfo{author}{\bibfnamefont{M.}~\bibnamefont{Horodecki}} \bibnamefont{and}
  \bibinfo{author}{\bibfnamefont{J.}~\bibnamefont{Oppenheim}},
  \bibinfo{journal}{International Journal of Modern Physics B}
  \textbf{\bibinfo{volume}{27}}, \bibinfo{pages}{1345019}
  (\bibinfo{year}{2013}).

\bibitem[{\citenamefont{Streltsov et~al.}(2015)\citenamefont{Streltsov, Singh,
  Dhar, Bera, and Adesso}}]{streltsov2015measuring}
\bibinfo{author}{\bibfnamefont{A.}~\bibnamefont{Streltsov}},
  \bibinfo{author}{\bibfnamefont{U.}~\bibnamefont{Singh}},
  \bibinfo{author}{\bibfnamefont{H.~S.} \bibnamefont{Dhar}},
  \bibinfo{author}{\bibfnamefont{M.~N.} \bibnamefont{Bera}}, \bibnamefont{and}
  \bibinfo{author}{\bibfnamefont{G.}~\bibnamefont{Adesso}},
  \bibinfo{journal}{Physical Review Letters} \textbf{\bibinfo{volume}{115}},
  \bibinfo{pages}{020403} (\bibinfo{year}{2015}).

\bibitem[{\citenamefont{Kwon et~al.}(2018)\citenamefont{Kwon, Jeong, Jennings,
  Yadin, and Kim}}]{kwon2018clock}
\bibinfo{author}{\bibfnamefont{H.}~\bibnamefont{Kwon}},
  \bibinfo{author}{\bibfnamefont{H.}~\bibnamefont{Jeong}},
  \bibinfo{author}{\bibfnamefont{D.}~\bibnamefont{Jennings}},
  \bibinfo{author}{\bibfnamefont{B.}~\bibnamefont{Yadin}}, \bibnamefont{and}
  \bibinfo{author}{\bibfnamefont{M.}~\bibnamefont{Kim}},
  \bibinfo{journal}{Physical Review Letters} \textbf{\bibinfo{volume}{120}},
  \bibinfo{pages}{150602} (\bibinfo{year}{2018}).

\bibitem[{\citenamefont{Oszmaniec et~al.}(2017)\citenamefont{Oszmaniec,
  Guerini, Wittek, and Ac{\'\i}n}}]{oszmaniec2017simulating}
\bibinfo{author}{\bibfnamefont{M.}~\bibnamefont{Oszmaniec}},
  \bibinfo{author}{\bibfnamefont{L.}~\bibnamefont{Guerini}},
  \bibinfo{author}{\bibfnamefont{P.}~\bibnamefont{Wittek}}, \bibnamefont{and}
  \bibinfo{author}{\bibfnamefont{A.}~\bibnamefont{Ac{\'\i}n}},
  \bibinfo{journal}{Physical Review Letters} \textbf{\bibinfo{volume}{119}},
  \bibinfo{pages}{190501} (\bibinfo{year}{2017}).

\bibitem[{\citenamefont{Bischof et~al.}(2018)\citenamefont{Bischof, Kampermann,
  and Bru{\ss}}}]{bischof2018resource}
\bibinfo{author}{\bibfnamefont{F.}~\bibnamefont{Bischof}},
  \bibinfo{author}{\bibfnamefont{H.}~\bibnamefont{Kampermann}},
  \bibnamefont{and} \bibinfo{author}{\bibfnamefont{D.}~\bibnamefont{Bru{\ss}}},
  \bibinfo{journal}{arXiv:1812.00018}  (\bibinfo{year}{2018}).

\bibitem[{\citenamefont{\r{A}berg}(2006)}]{aberg2006quantifying}
\bibinfo{author}{\bibfnamefont{J.}~\bibnamefont{\r{A}berg}},
  \bibinfo{journal}{arXiv:quant-ph/0612146}  (\bibinfo{year}{2006}).

\bibitem[{\citenamefont{Piani et~al.}(2016)\citenamefont{Piani, Cianciaruso,
  Bromley, Napoli, Johnston, and Adesso}}]{piani2016robustness}
\bibinfo{author}{\bibfnamefont{M.}~\bibnamefont{Piani}},
  \bibinfo{author}{\bibfnamefont{M.}~\bibnamefont{Cianciaruso}},
  \bibinfo{author}{\bibfnamefont{T.~R.} \bibnamefont{Bromley}},
  \bibinfo{author}{\bibfnamefont{C.}~\bibnamefont{Napoli}},
  \bibinfo{author}{\bibfnamefont{N.}~\bibnamefont{Johnston}}, \bibnamefont{and}
  \bibinfo{author}{\bibfnamefont{G.}~\bibnamefont{Adesso}},
  \bibinfo{journal}{Physical Review A} \textbf{\bibinfo{volume}{93}},
  \bibinfo{pages}{042107} (\bibinfo{year}{2016}).

\bibitem[{\citenamefont{Gour et~al.}(2009)\citenamefont{Gour, Marvian, and
  Spekkens}}]{gour2009measuring}
\bibinfo{author}{\bibfnamefont{G.}~\bibnamefont{Gour}},
  \bibinfo{author}{\bibfnamefont{I.}~\bibnamefont{Marvian}}, \bibnamefont{and}
  \bibinfo{author}{\bibfnamefont{R.~W.} \bibnamefont{Spekkens}},
  \bibinfo{journal}{Physical Review A} \textbf{\bibinfo{volume}{80}},
  \bibinfo{pages}{012307} (\bibinfo{year}{2009}).

\bibitem[{\citenamefont{Marvian et~al.}(2016)\citenamefont{Marvian, Spekkens,
  and Zanardi}}]{marvian2016quantum}
\bibinfo{author}{\bibfnamefont{I.}~\bibnamefont{Marvian}},
  \bibinfo{author}{\bibfnamefont{R.~W.} \bibnamefont{Spekkens}},
  \bibnamefont{and} \bibinfo{author}{\bibfnamefont{P.}~\bibnamefont{Zanardi}},
  \bibinfo{journal}{Physical Review A} \textbf{\bibinfo{volume}{93}},
  \bibinfo{pages}{052331} (\bibinfo{year}{2016}).

\bibitem[{\citenamefont{Marvian and Spekkens}(2016)}]{marvian2016quantify}
\bibinfo{author}{\bibfnamefont{I.}~\bibnamefont{Marvian}} \bibnamefont{and}
  \bibinfo{author}{\bibfnamefont{R.~W.} \bibnamefont{Spekkens}},
  \bibinfo{journal}{Physical Review A} \textbf{\bibinfo{volume}{94}},
  \bibinfo{pages}{052324} (\bibinfo{year}{2016}).

\bibitem[{\citenamefont{D'Ariano et~al.}(2005)\citenamefont{D'Ariano, Presti,
  and Perinotti}}]{d2005classical}
\bibinfo{author}{\bibfnamefont{G.~M.} \bibnamefont{D'Ariano}},
  \bibinfo{author}{\bibfnamefont{P.~L.} \bibnamefont{Presti}},
  \bibnamefont{and}
  \bibinfo{author}{\bibfnamefont{P.}~\bibnamefont{Perinotti}},
  \bibinfo{journal}{Journal of Physics A: Mathematical and General}
  \textbf{\bibinfo{volume}{38}}, \bibinfo{pages}{5979} (\bibinfo{year}{2005}).

\bibitem[{\citenamefont{Yuan et~al.}(2015)\citenamefont{Yuan, Zhou, Cao, and
  Ma}}]{yuan2015intrinsic}
\bibinfo{author}{\bibfnamefont{X.}~\bibnamefont{Yuan}},
  \bibinfo{author}{\bibfnamefont{H.}~\bibnamefont{Zhou}},
  \bibinfo{author}{\bibfnamefont{Z.}~\bibnamefont{Cao}}, \bibnamefont{and}
  \bibinfo{author}{\bibfnamefont{X.}~\bibnamefont{Ma}},
  \bibinfo{journal}{Physical Review A} \textbf{\bibinfo{volume}{92}},
  \bibinfo{pages}{022124} (\bibinfo{year}{2015}).

\bibitem[{\citenamefont{Yuan et~al.}(2016)\citenamefont{Yuan, Zhao, Girolami,
  and Ma}}]{yuan2016interplay}
\bibinfo{author}{\bibfnamefont{X.}~\bibnamefont{Yuan}},
  \bibinfo{author}{\bibfnamefont{Q.}~\bibnamefont{Zhao}},
  \bibinfo{author}{\bibfnamefont{D.}~\bibnamefont{Girolami}}, \bibnamefont{and}
  \bibinfo{author}{\bibfnamefont{X.}~\bibnamefont{Ma}},
  \bibinfo{journal}{arXiv:1605.07818}  (\bibinfo{year}{2016}).

\bibitem[{\citenamefont{Tomamichel et~al.}(2009)\citenamefont{Tomamichel,
  Colbeck, and Renner}}]{tomamichel2009fully}
\bibinfo{author}{\bibfnamefont{M.}~\bibnamefont{Tomamichel}},
  \bibinfo{author}{\bibfnamefont{R.}~\bibnamefont{Colbeck}}, \bibnamefont{and}
  \bibinfo{author}{\bibfnamefont{R.}~\bibnamefont{Renner}},
  \bibinfo{journal}{IEEE Transactions on Information Theory}
  \textbf{\bibinfo{volume}{55}}, \bibinfo{pages}{5840} (\bibinfo{year}{2009}).

\bibitem[{\citenamefont{Renner}(2008)}]{renner2008security}
\bibinfo{author}{\bibfnamefont{R.}~\bibnamefont{Renner}},
  \bibinfo{journal}{International Journal of Quantum Information}
  \textbf{\bibinfo{volume}{6}}, \bibinfo{pages}{1} (\bibinfo{year}{2008}).

\bibitem[{\citenamefont{Nielsen and Chuang}(2000)}]{nielsenchuang}
\bibinfo{author}{\bibfnamefont{M.~A.} \bibnamefont{Nielsen}} \bibnamefont{and}
  \bibinfo{author}{\bibfnamefont{I.~L.} \bibnamefont{Chuang}},
  \emph{\bibinfo{title}{Quantum Computation and Quantum Information}}
  (\bibinfo{publisher}{Cambridge University Press}, \bibinfo{year}{2000}).

\bibitem[{\citenamefont{Cao et~al.}(2015)\citenamefont{Cao, Zhou, and
  Ma}}]{ma15}
\bibinfo{author}{\bibfnamefont{Z.}~\bibnamefont{Cao}},
  \bibinfo{author}{\bibfnamefont{H.}~\bibnamefont{Zhou}}, \bibnamefont{and}
  \bibinfo{author}{\bibfnamefont{X.}~\bibnamefont{Ma}}, \bibinfo{journal}{New
  Journal of Physics} \textbf{\bibinfo{volume}{17}}, \bibinfo{pages}{125011}
  (\bibinfo{year}{2015}).

\bibitem[{\citenamefont{Bischof et~al.}(2017)\citenamefont{Bischof, Kampermann,
  and Bru{\ss}}}]{bischof2017measurement}
\bibinfo{author}{\bibfnamefont{F.}~\bibnamefont{Bischof}},
  \bibinfo{author}{\bibfnamefont{H.}~\bibnamefont{Kampermann}},
  \bibnamefont{and} \bibinfo{author}{\bibfnamefont{D.}~\bibnamefont{Bru{\ss}}},
  \bibinfo{journal}{Physical Review A} \textbf{\bibinfo{volume}{95}},
  \bibinfo{pages}{062305} (\bibinfo{year}{2017}).

\bibitem[{\citenamefont{Brask et~al.}(2017)\citenamefont{Brask, Martin,
  Esposito, Houlmann, Bowles, Zbinden, and Brunner}}]{brask2017megahertz}
\bibinfo{author}{\bibfnamefont{J.~B.} \bibnamefont{Brask}},
  \bibinfo{author}{\bibfnamefont{A.}~\bibnamefont{Martin}},
  \bibinfo{author}{\bibfnamefont{W.}~\bibnamefont{Esposito}},
  \bibinfo{author}{\bibfnamefont{R.}~\bibnamefont{Houlmann}},
  \bibinfo{author}{\bibfnamefont{J.}~\bibnamefont{Bowles}},
  \bibinfo{author}{\bibfnamefont{H.}~\bibnamefont{Zbinden}}, \bibnamefont{and}
  \bibinfo{author}{\bibfnamefont{N.}~\bibnamefont{Brunner}},
  \bibinfo{journal}{Physical Review Applied} \textbf{\bibinfo{volume}{7}},
  \bibinfo{pages}{054018} (\bibinfo{year}{2017}).

\bibitem[{\citenamefont{Ioannou et~al.}(2019)\citenamefont{Ioannou, Brask, and
  Brunner}}]{ioannou2019upper}
\bibinfo{author}{\bibfnamefont{M.}~\bibnamefont{Ioannou}},
  \bibinfo{author}{\bibfnamefont{J.~B.} \bibnamefont{Brask}}, \bibnamefont{and}
  \bibinfo{author}{\bibfnamefont{N.}~\bibnamefont{Brunner}},
  \bibinfo{journal}{Physical Review A} \textbf{\bibinfo{volume}{99}},
  \bibinfo{pages}{052338} (\bibinfo{year}{2019}).

\bibitem[{\citenamefont{Napoli et~al.}(2016)\citenamefont{Napoli, Bromley,
  Cianciaruso, Piani, Johnston, and Adesso}}]{napoli2016robustness}
\bibinfo{author}{\bibfnamefont{C.}~\bibnamefont{Napoli}},
  \bibinfo{author}{\bibfnamefont{T.~R.} \bibnamefont{Bromley}},
  \bibinfo{author}{\bibfnamefont{M.}~\bibnamefont{Cianciaruso}},
  \bibinfo{author}{\bibfnamefont{M.}~\bibnamefont{Piani}},
  \bibinfo{author}{\bibfnamefont{N.}~\bibnamefont{Johnston}}, \bibnamefont{and}
  \bibinfo{author}{\bibfnamefont{G.}~\bibnamefont{Adesso}},
  \bibinfo{journal}{Physical Review Letters} \textbf{\bibinfo{volume}{116}},
  \bibinfo{pages}{150502} (\bibinfo{year}{2016}).

\bibitem[{\citenamefont{Bu et~al.}(2017)\citenamefont{Bu, Singh, Fei, Pati, and
  Wu}}]{bu2017maximum}
\bibinfo{author}{\bibfnamefont{K.}~\bibnamefont{Bu}},
  \bibinfo{author}{\bibfnamefont{U.}~\bibnamefont{Singh}},
  \bibinfo{author}{\bibfnamefont{S.-M.} \bibnamefont{Fei}},
  \bibinfo{author}{\bibfnamefont{A.~K.} \bibnamefont{Pati}}, \bibnamefont{and}
  \bibinfo{author}{\bibfnamefont{J.}~\bibnamefont{Wu}},
  \bibinfo{journal}{Physical Review Letters} \textbf{\bibinfo{volume}{119}},
  \bibinfo{pages}{150405} (\bibinfo{year}{2017}).

\bibitem[{\citenamefont{L{\"{o}}fberg}(2004)}]{Lofberg2004}
\bibinfo{author}{\bibfnamefont{J.}~\bibnamefont{L{\"{o}}fberg}}, in
  \emph{\bibinfo{booktitle}{In Proceedings of the CACSD Conference}}
  (\bibinfo{address}{Taipei, Taiwan}, \bibinfo{year}{2004}).

\bibitem[{\citenamefont{Bu and Xiong}(2017)}]{bu2016note}
\bibinfo{author}{\bibfnamefont{K.}~\bibnamefont{Bu}} \bibnamefont{and}
  \bibinfo{author}{\bibfnamefont{C.}~\bibnamefont{Xiong}},
  \bibinfo{journal}{Quantum Information \& Computation}
  \textbf{\bibinfo{volume}{17}}, \bibinfo{pages}{1206} (\bibinfo{year}{2017}).

\bibitem[{\citenamefont{Rana et~al.}(2017)\citenamefont{Rana, Parashar, Winter,
  and Lewenstein}}]{rana2017logarithmic}
\bibinfo{author}{\bibfnamefont{S.}~\bibnamefont{Rana}},
  \bibinfo{author}{\bibfnamefont{P.}~\bibnamefont{Parashar}},
  \bibinfo{author}{\bibfnamefont{A.}~\bibnamefont{Winter}}, \bibnamefont{and}
  \bibinfo{author}{\bibfnamefont{M.}~\bibnamefont{Lewenstein}},
  \bibinfo{journal}{Physical Review A} \textbf{\bibinfo{volume}{96}},
  \bibinfo{pages}{052336} (\bibinfo{year}{2017}).

\bibitem[{\citenamefont{Chitambar and
  Gour}(2016{\natexlab{a}})}]{chitambar2016comparison}
\bibinfo{author}{\bibfnamefont{E.}~\bibnamefont{Chitambar}} \bibnamefont{and}
  \bibinfo{author}{\bibfnamefont{G.}~\bibnamefont{Gour}},
  \bibinfo{journal}{Physical Review A} \textbf{\bibinfo{volume}{94}},
  \bibinfo{pages}{052336} (\bibinfo{year}{2016}{\natexlab{a}}).

\bibitem[{\citenamefont{Chitambar and
  Gour}(2016{\natexlab{b}})}]{chitambar2016critical}
\bibinfo{author}{\bibfnamefont{E.}~\bibnamefont{Chitambar}} \bibnamefont{and}
  \bibinfo{author}{\bibfnamefont{G.}~\bibnamefont{Gour}},
  \bibinfo{journal}{Physical Review Letters} \textbf{\bibinfo{volume}{117}},
  \bibinfo{pages}{030401} (\bibinfo{year}{2016}{\natexlab{b}}).

\bibitem[{\citenamefont{Yadin et~al.}(2016)\citenamefont{Yadin, Ma, Girolami,
  Gu, and Vedral}}]{yadin2016quantum}
\bibinfo{author}{\bibfnamefont{B.}~\bibnamefont{Yadin}},
  \bibinfo{author}{\bibfnamefont{J.}~\bibnamefont{Ma}},
  \bibinfo{author}{\bibfnamefont{D.}~\bibnamefont{Girolami}},
  \bibinfo{author}{\bibfnamefont{M.}~\bibnamefont{Gu}}, \bibnamefont{and}
  \bibinfo{author}{\bibfnamefont{V.}~\bibnamefont{Vedral}},
  \bibinfo{journal}{Physical Review X} \textbf{\bibinfo{volume}{6}},
  \bibinfo{pages}{041028} (\bibinfo{year}{2016}).

\bibitem[{\citenamefont{Zhao et~al.}(2018)\citenamefont{Zhao, Liu, Yuan,
  Chitambar, and Ma}}]{zhao2018one}
\bibinfo{author}{\bibfnamefont{Q.}~\bibnamefont{Zhao}},
  \bibinfo{author}{\bibfnamefont{Y.}~\bibnamefont{Liu}},
  \bibinfo{author}{\bibfnamefont{X.}~\bibnamefont{Yuan}},
  \bibinfo{author}{\bibfnamefont{E.}~\bibnamefont{Chitambar}},
  \bibnamefont{and} \bibinfo{author}{\bibfnamefont{X.}~\bibnamefont{Ma}},
  \bibinfo{journal}{Physical Review Letters} \textbf{\bibinfo{volume}{120}},
  \bibinfo{pages}{070403} (\bibinfo{year}{2018}).

\bibitem[{\citenamefont{Regula et~al.}(2018)\citenamefont{Regula, Fang, Wang,
  and Adesso}}]{regula2018one}
\bibinfo{author}{\bibfnamefont{B.}~\bibnamefont{Regula}},
  \bibinfo{author}{\bibfnamefont{K.}~\bibnamefont{Fang}},
  \bibinfo{author}{\bibfnamefont{X.}~\bibnamefont{Wang}}, \bibnamefont{and}
  \bibinfo{author}{\bibfnamefont{G.}~\bibnamefont{Adesso}},
  \bibinfo{journal}{Physical review letters} \textbf{\bibinfo{volume}{121}},
  \bibinfo{pages}{010401} (\bibinfo{year}{2018}).

\bibitem[{\citenamefont{Lami}(2019)}]{lami2019completing}
\bibinfo{author}{\bibfnamefont{L.}~\bibnamefont{Lami}},
  \bibinfo{journal}{arXiv:1902.02427}  (\bibinfo{year}{2019}).

\bibitem[{\citenamefont{Zhao et~al.}(2019)\citenamefont{Zhao, Liu, Yuan,
  Chitambar, and Winter}}]{zhao2019one}
\bibinfo{author}{\bibfnamefont{Q.}~\bibnamefont{Zhao}},
  \bibinfo{author}{\bibfnamefont{Y.}~\bibnamefont{Liu}},
  \bibinfo{author}{\bibfnamefont{X.}~\bibnamefont{Yuan}},
  \bibinfo{author}{\bibfnamefont{E.}~\bibnamefont{Chitambar}},
  \bibnamefont{and} \bibinfo{author}{\bibfnamefont{A.}~\bibnamefont{Winter}},
  \bibinfo{journal}{IEEE Transactions on Information Theory}
  (\bibinfo{year}{2019}).

\bibitem[{\citenamefont{Lami et~al.}(2019)\citenamefont{Lami, Regula, and
  Adesso}}]{lami2019generic}
\bibinfo{author}{\bibfnamefont{L.}~\bibnamefont{Lami}},
  \bibinfo{author}{\bibfnamefont{B.}~\bibnamefont{Regula}}, \bibnamefont{and}
  \bibinfo{author}{\bibfnamefont{G.}~\bibnamefont{Adesso}},
  \bibinfo{journal}{Physical Review Letters} \textbf{\bibinfo{volume}{122}},
  \bibinfo{pages}{150402} (\bibinfo{year}{2019}).

\bibitem[{\citenamefont{Streltsov et~al.}(2018)\citenamefont{Streltsov,
  Kampermann, W{\"o}lk, Gessner, and Bru{\ss}}}]{streltsov2018maximal}
\bibinfo{author}{\bibfnamefont{A.}~\bibnamefont{Streltsov}},
  \bibinfo{author}{\bibfnamefont{H.}~\bibnamefont{Kampermann}},
  \bibinfo{author}{\bibfnamefont{S.}~\bibnamefont{W{\"o}lk}},
  \bibinfo{author}{\bibfnamefont{M.}~\bibnamefont{Gessner}}, \bibnamefont{and}
  \bibinfo{author}{\bibfnamefont{D.}~\bibnamefont{Bru{\ss}}},
  \bibinfo{journal}{New Journal of Physics} \textbf{\bibinfo{volume}{20}},
  \bibinfo{pages}{053058} (\bibinfo{year}{2018}).

\bibitem[{\citenamefont{Watrous}(2018)}]{watrous2018theory}
\bibinfo{author}{\bibfnamefont{J.}~\bibnamefont{Watrous}},
  \emph{\bibinfo{title}{The theory of quantum information}}
  (\bibinfo{publisher}{Cambridge University Press}, \bibinfo{year}{2018}).

\end{thebibliography}
\bibliographystyle{apsrev}

\end{document}